\documentclass[aps,prd,showpacs,amssymb,superscriptaddress,nofootinbib,twocolumn]{revtex4-1}

\usepackage{amsmath,amsfonts,amsthm,amssymb}
\usepackage{braket}
\usepackage{setspace}
\usepackage{amsmath}
\usepackage{appendix}
\usepackage[ansinew]{inputenc}
\usepackage{bbm}
\usepackage{bm}
\usepackage{amsbsy}
\usepackage{dsfont} % for symbols like the identity: \mathds{1}
\usepackage{graphicx} % for graphics
\usepackage{epsfig}
\usepackage{epstopdf}
\usepackage{dsfont}
\usepackage{color}
\usepackage[colorlinks]{hyperref}
\usepackage[figure,table]{hypcap}
\usepackage{enumerate}%allows different styles of enumerate environment
\usepackage[resetlabels]{multibib}

\hypersetup{
	bookmarksnumbered,
	pdfstartview={FitH},
	citecolor={darkgreen},
	linkcolor={darkred},
	urlcolor={darkblue},
	pdfpagemode={UseOutlines}}
\definecolor{darkgreen}{RGB}{50,190,50}
\definecolor{darkblue}{RGB}{0,0,190}
\definecolor{darkred}{RGB}{238,0,0}
\usepackage{soul}
\newcommand{\nl}{\ensuremath{\hspace*{-0.5pt}}}
\newcommand{\nr}{\ensuremath{\hspace*{0.5pt}}}

\newcommand{\sigA}{\ensuremath{\sigma^{\hspace{-0.0pt}\protect\raisebox{0pt}{\tiny{$A$}}}}}
\newcommand{\sigB}{\ensuremath{\sigma^{\hspace{-0.0pt}\protect\raisebox{0pt}{\tiny{$B$}}}}}
\newcommand{\sigi}{\ensuremath{\sigma^{\hspace{-0.0pt}\protect\raisebox{0pt}{\tiny{$i$}}}}}
\newcommand{\vecsigA}{\ensuremath{\vec{\sigma}^{\hspace{0pt}\protect\raisebox{0pt}{\tiny{$A$}}}}}
\newcommand{\vecsigB}{\ensuremath{\vec{\sigma}^{\hspace{0pt}\protect\raisebox{0pt}{\tiny{$B$}}}}}

\newcommand{\expval}[1]{\ensuremath{\left\langle\right. #1 \left.\right\rangle}}
\newcommand{\HSpr}[2]{\ensuremath{\protect{(\nr #1 \nr|\nr #2 \nr)\subtiny{0}{-1}{\rm{HS}}}}}

\newcommand{\pr}{^{\prime}}

\newcommand{\rhoA}[1]{\ensuremath{\rho_{\hspace*{-1.0pt}\protect\raisebox{-1.0pt}{\tiny{$ #1 $}}}}}

\newcommand{\subtiny}[3]{\ensuremath{_{\hspace{#1 pt}\protect\raisebox{#2 pt}{\tiny{$ #3$}}}}}
\newcommand{\suptiny}[3]{\ensuremath{^{\hspace{#1 pt}\protect\raisebox{#2 pt}{\tiny{$ #3$}}}}}

\newcommand{\subA}[1]{\ensuremath{_{\hspace*{-1.0pt}\protect\raisebox{-1.0pt}{\tiny{$ #1 $}}}}}
\newcommand{\supA}[1]{\ensuremath{^{\hspace*{-0.0pt}\protect\raisebox{-1.0pt}{\scriptsize{$ #1 $}}}}}
\DeclareMathOperator{\artanh}{artanh}

\DeclareMathOperator{\diag}{diag}

\newcommand{\tr}{\textnormal{Tr}}
\newcommand{\djj}{d\kern-0.4em\char"16\kern-0.1em}
\newcommand{\beq}{\begin{eqnarray}}
\newcommand{\eeq}{\end{eqnarray}}

\newcommand{\B}{\ensuremath{\protect{\mathcal{B}}}}
\newcommand{\Ha}{\mathcal{H}}

\newtheorem{theorem}{Theorem}[section]

\begin{document}

\title{Geometry of two-qubit states with negative conditional entropy}
\author{Nicolai Friis}
\email{nicolai.friis@uibk.ac.at}
\affiliation{
Institute for Theoretical Physics, University of Innsbruck,
Technikerstra{\ss}e 21a,
A-6020 Innsbruck,
Austria}
\author{Sridhar Bulusu}
%\email{}
\affiliation{Faculty of Physics, University of Vienna, Boltzmanngasse 5, 1090 Vienna, Austria}
\author{Reinhold A. Bertlmann}
\email{reinhold.bertlmann@univie.ac.at}
\affiliation{Faculty of Physics, University of Vienna, Boltzmanngasse 5, 1090 Vienna, Austria}

\begin{abstract}
We review the geometric features of negative conditional entropy and the properties of the conditional amplitude operator proposed by Cerf and Adami for two qubit states in comparison with entanglement and nonlocality of the states. We identify the region of negative conditional entropy in the tetrahedron of locally maximally mixed two-qubit states. Within this set of states, negative conditional entropy implies nonlocality and entanglement, but not vice versa, and we show that the Cerf-Adami conditional amplitude operator provides an entanglement witness equivalent to the Peres-Horodecki criterion. Outside of the tetrahedron this equivalence is generally not true.
\end{abstract}


\maketitle

\section{Introduction}\label{sec:introduction}

The feature of entanglement is the basis for many fascinating phenomena in quantum information and quantum communication, such as quantum teleportation~\cite{BennettBrassardCrepeauJoszaPeresWootters1993,HorodeckiMPR1999} or quantum cryptography~\cite{BennettBrassard1984,Ekert1991,GisinRibordyTittelZbinden2002}. Although the division of quantum states into entangled and separable states is well-defined mathematically, checking whether a given state is entangled or not often proves to be extraordinarily difficult. Consequently, a plethora of inequivalent criteria and measures is available for the detection and classification of entanglement~\cite{PlenioVirmani2007,HorodeckiRPMK2007,GuehneToth2009}. The spectrum of available methods ranges from entanglement monotones such as the concurrence~\cite{Hill-Wootters1997,Wootters1998,Wootters2001} or negativity~\cite{Plenio2005}, and geometric entanglement detection criteria in the form of so-called entanglement witnesses~\cite{HorodeckiMPR1996, Terhal2000,BertlmannNarnhoferThirring2002,Bruss2002}, to measures that directly quantify the utility of a state for specific tasks requiring entanglement. One of the most prominent such tasks, which challenges our preconceptions of the reality of nature~\cite{Bell1981}, is the test of~a Bell inequality~\cite{Bell1964, Bell:Speakable1987}, distinguishing between so-called local and nonlocal states, for which the inequality is satisfied or violated, respectively. However, while entanglement is required to violate a Bell inequality, entanglement and nonlocality are not the same concepts. As discovered by Werner~\cite{Werner1989}, certain mixed states, albeit still being entangled, cannot be used to violate a Bell inequality and hence behave like strictly local states.

In contrast to methods that directly relate entanglement to physically measurable quantities stand information-theoretic approaches based on the entropies of quantum states. In both classical and quantum information theory entropies play a crucial role. Quite generally, entropy represents the degree of uncertainty \textemdash\ the lack of knowledge \textemdash\ about a (quantum) system. More specifically, the von Neumann entropy of a quantum state can be interpreted~\cite{Schumacher1995, JoszaSchumacher1994} as the minimal amount of information necessary to fully specify the state, be it separable or entangled. For the quantification of the correlations between two subsystems $A$ and $B$, two particularly interesting entropies are the mutual entropy (or mutual information) $S(A:\!B)$ and the conditional entropy $S(A|B)$. In analogy to the classical case, the mutual entropy $S(A:\!B)$ corresponds to the amount of information contained in the joint state that exceeds the information locally available to $A$ and $B$, i.e., $S(A:\!B)$ is a measure for the degree of correlation between subsystems $A$ and $B\,$. On the other hand, $S(A|B)$ is the entropy of the state of subsystem $A$ conditioned on the knowledge of the state of subsystem $B$. In a series of papers~\cite{Cerf-Adami-PRL1997, Cerf-Adami-PRA1999, Cerf-Adami-9605002v2, CerfAdami1998, CerfAdami1997} investigating the conditional entropy and the mutual entropy by means of so-called \emph{mutual} and \emph{conditional amplitude operators} (CAO), Cerf and Adami concluded that the quantum conditional entropy \textemdash\ in contrast to its classical counterpart \textemdash\ can become negative for entangled states. This provides a connection between quantum nonseparability and conditional entropy, or mutual entropy that we wish to investigate further in this article.

The purpose of our article is hence to review the geometry of quantum states with negative conditional entropy and to compare it with the different regions of nonlocality, entanglement and separability. In particular, we want to focus on the paradigmatic case of two qubits, which is one of the very few examples where the different methods for detection and quantification of entanglement and nonlocality described above are practically computable and can be compared both numerically and geometrically. In this sense, albeit being a system of comparatively small complexity, the two-qubit case is of high significance, since it serves as a guiding example for developing the geometric understanding and intuition necessary to study more complicated systems.

Our investigation confirms that for the interesting class of locally maximally mixed states, the requirement of negative conditional entropy is a strictly stronger constraint than that of nonlocality, i.e., all states with negative conditional entropy are nonlocal, and therefore entangled, but the converse statements do not hold. We then consider an entanglement criterion based on the Cerf-Adami conditional amplitude operator and show that it is equivalent to the Peres-Horodecki criterion~\cite{Peres1996,HorodeckiMPR1996} for the set of locally maximally mixed states, but not for all two-qubit states.

This paper is structured as follows. We begin with a pedagogical review of the basic methods in Sect.~\ref{sec:methods}, discussing the geometric entanglement and separability characteristics in Sects.~\ref{sec:entanglement-separability} and~\ref{sec:entanglement witness}, the boundary between local and nonlocal states in Sects.~\ref{sec:nonlocality} and~\ref{sec:hidden nonlocality}, and the entropic correlation measures in Sects.~\ref{sec:entropy measures} and~\ref{sec:entropy}. We then present the results of our investigation in Sect.~\ref{sec:results}, where we discuss the geometric aspects of the conditional entropy within the set of locally maximally mixed states in Sect.~\ref{sec:geometry-two-qubit-states}, provide an example for the general inequivalence of the Peres-Horodecki and the CAO criterion in Sect.~\ref{sec:counterexample}, and discuss extensions to generalized entropies in Sect.~\ref{sec:Negativity of Generalized Conditional Entropies}. Finally, we draw conclusions in Sect.~\ref{sec:conclusion}.\\
\vspace*{-7mm}

\section{Methods}\label{sec:methods}
\vspace*{-1mm}

In this section we will provide a pedagogical review of the methods for the detection and quantification of entanglement and nonlocality relevant to this study. The reader already well familiar with the geometry of separable, entangled, and nonlocal states for two-qubits may skip directly to Sect.~\ref{sec:results}, where we present our results.

\vspace*{-1.5mm}
\subsection{Entanglement \& Separability}\label{sec:entanglement-separability}
\vspace*{-1.5mm}

Quantum states are described by density operators $\rho$, i.e., positive semi-definite ($\rho\geq0$), hermitean ($\rho=\rho^{\dagger}$, which of course follows from positivity) operators with unit trace, $\tr(\rho)=1$. These operators form a convex subset $\widetilde{\mathcal{H}}\subset\mathcal{L}(\Ha)$ in the Hilbert-Schmidt space $\mathcal{L}(\Ha)$ of linear operators over the Hilbert space $\Ha$ of pure states. Given a bipartition of the Hilbert space into two subsystems $A$ and $B$ with respect to the tensor product, $\mathcal{H}=\Ha\subA{A}\otimes\Ha\subA{B}$, one may classify the quantum states into separable and entangled states. The \emph{set} $\mathcal{S}$ \emph{of separable states} is defined by the convex (and compact) hull of product states
\begin{align}
    \hspace*{-1mm}\mathcal{S} &= \Big\{\rho = \sum_{n} p_{n} \,\rho\supA{A}_{n} \otimes \rho\supA{B}_{n} \,|\; 0\leq p_{n}\leq 1\,, \sum_{n} p_{n} =1  \Big\},
    \label{set-separable-states}
\end{align}
where $\rho\supA{A}_{n}$ and $\rho\supA{B}_{n}$ are density operators in $\widetilde{\Ha}\subA{A}\subset\mathcal{L}(\mathcal{H}\subA{A})$ and $\widetilde{\Ha}\subA{B}\subset\mathcal{L}(\mathcal{H}\subA{B})$, respectively. In contrast, any state that is not separable, i.e., which cannot be expressed as a convex combination of product states, is called \emph{entangled}. The set $\mathcal{S}^{c}$ of entangled states hence forms the complement of the set of separable states, such that $\mathcal{S} \cup \mathcal{S}^c = \widetilde{\Ha}$.

Here we would like to emphasize that the characterization of a given state as being entangled or separable very much depends on the choice of factorizing the algebra of the corresponding density matrix~\cite{ThirringBertlmannKoehlerNarnhofer2011,Zanardi2001}. From a practical point of view, this choice of bipartition is often suggested by the experimental setup, e.g., by the spatial separation of the observers Alice and Bob corresponding to subsystems $A$ and $B$, respectively. From the perspective of a theorist on the other hand, one has a freedom to choose the bipartition into two subsystems. While a given density operator may well be separable with respect to the decomposition $\widetilde{\Ha}=\widetilde{\Ha}\subA{A}\otimes\widetilde{\Ha}\subA{B}$, it may be entangled with respect to another factorization $\widetilde{\Ha}\subA{A\pr}\otimes\widetilde{\Ha}\subA{B\pr}$. Since such a different choice of bipartition corresponds to a change of basis in $\Ha$, it can be represented by a (global) unitary transformation. As shown in Ref.~\cite{ThirringBertlmannKoehlerNarnhofer2011}, every separable pure state admits a unitary operator transforming it to an entangled state, and vice versa. Interestingly, for mixed states this switch between separability and entanglement is only possible above a certain bound of purity. This implies that there exist quantum states which are separable with respect to all possible factorizations of the composite system into subsystems. This is the case if $U\rho\,U^\dag$ remains separable for any unitary transformation $U$. Such states are called \emph{absolutely separable states}~\cite{kus-zyczkowski2001, Zyczkowski-Bengtsson2006, Bengtsson-Zyczkowski-Book2007}. Geometrically one may think of the absolutely separable states as a convex and compact~\cite{GangulyChatterjeeMajumdar2014} subset $\mathcal{A}$ of the separable states $\mathcal{S}$, much like $\mathcal{S}$ forms a convex subset of all states. In particular, when $\dim(\widetilde{\Ha}\subA{A})=\dim(\widetilde{\Ha}\subA{B})=d$, one may inscribe a ball of maximal radius $r_{\mathrm{max}}=1/\bigl(d\sqrt{d^{2}-1}\bigr)$ into the set $\mathcal{S}$, where the distance of a state $\rho$ from the central maximally mixed state $\rho_{\mathrm{mix}}=\mathds{1}_{d^{2}}/d^{2}$ is measured by the trace distance $|\hspace*{-0.5pt}|\rho-\rho_{\mathrm{mix}}|\hspace*{-0.5pt}|=\sqrt{\tr(\rho-\rho_{\mathrm{mix}})^{2}}$. All states within this so-called  Ku\'s-\.Zyczkowski ball~\cite{kus-zyczkowski2001} are separable. Moreover, since the condition $r\leq r_{\mathrm{max}}$ translates to the purity as $\tr(\rho^{2})\leq1/\bigl(d^{2}-1\bigr)$ and (global) unitaries leave the purity invariant, all states within this maximal ball are also absolutely separable. However, note that not all absolutely separable states lie within this ball, see, e.g., Refs.~\cite{ThirringBertlmannKoehlerNarnhofer2011, kus-zyczkowski2001, GurvitsBarnum2002}.

\vspace*{-1mm}
The convex nesting hierarchy $\mathcal{A}\subset\mathcal{S}\subset\widetilde{\Ha}\subset\mathcal{L}(\Ha)$ holds for (bipartite) quantum systems of arbitrary dimensions $\dim(\widetilde{\Ha}\subA{A})=d\subA{A}$ and $\dim(\widetilde{\Ha}\subA{B})=d\subA{B}$. The density operators of such systems can be written in a generalized Bloch-Fano decomposition~\cite{Fano1983,BertlmannKrammer2008b} as
\vspace*{-2mm}
\begin{align}
    \rho    &=\,\frac{1}{d\subA{A}d\subA{B}}\Bigl(\mathds{1}\subA{AB}\,+\,
    \sum\limits_{m=1}^{d\subA{A}^{2}-1}a_{m}\,\sigA_{m}\otimes\mathds{1}\subA{B}
        +\,\sum\limits_{n=1}^{d\subA{B}^{2}-1}b_{n}\,\mathds{1}\subA{A}\!\otimes\sigB_{n}\nonumber\\
        &\qquad +\,\sum\limits_{m=1}^{d\subA{A}^{2}-1}\sum\limits_{n=1}^{d\subA{B}^{2}-1}\,t_{mn}\,\sigA_{m}\otimes\sigB_{n}\Bigr)\,,
    \label{eq:density operator bloch decomp}
\end{align}
where the hermitean operators $\sigi_{m}$ for $i=A,B$ are the generalizations of the Pauli matrices, i.e., they are orthogonal in the sense that $\tr(\sigi_{m}\sigi_{n})=2\delta_{mn}$, and traceless, $\tr(\sigi_{m})=0$, and they coincide with the Pauli matrices for dimension $2$. The coefficients $a_{m},b_{n}\in\mathbb{R}$ are the components of the generalized Bloch vectors $\vec{a}$ and $\vec{b}$ of the subsystems $A$ and $B$, respectively, which completely determine the reduced states $\rho\subA{A}=\tr\subA{B}(\rho)$ and $\rho\subA{B}=\tr\subA{A}(\rho)$. The real coefficients $t_{mn}$ are the components of the so-called correlation tensor. Note that $a_{m}$, $b_{n}$, and $t_{mn}$ cannot be chosen completely independently, but are jointly constrained by the positivity of~$\rho$.

An interesting subset of the state space is given by the set $\mathcal{W}$ of \emph{locally maximally mixed states} or \emph{Weyl states}, that is, the set of quantum systems with vanishing Bloch vectors, $a_{m}=b_{n}=0\ \forall m,n$ such that $\rho\subA{A}=\mathds{1}\subA{A}/d\subA{A}$ and $\rho\subA{B}=\mathds{1}\subA{B}/d\subA{B}$. The set $\mathcal{W}$ contains all the maximally entangled states (for which the marginals $\rho\subA{A}$ and $\rhoA{B}$ are maximally mixed) and the uncorrelated maximally mixed state $\rho=\mathds{1}\subA{AB}/(d\subA{A}d\subA{B})$, and all states in $\mathcal{W}$ are fully determined by their correlation tensors $t=(t_{mn})$. For $d\subA{A}=d\subA{B}=d$, the singular value decomposition of the correlation tensor allows bringing $t$ to a diagonal form $\tilde{t}=\diag\{\tilde{t}_{n}\}$ using two orthogonal transformations $R_{1}$ and $R_{2}$, such that $R_{1}t\hspace*{1pt}R_{2}=\tilde{t}$. Moreover, these orthogonal transformations can be realized by local unitaries $U_{1}\otimes U_{2}$, which do not change the entanglement (or the entropy) of the state. This means that, up to local unitaries, all Weyl states for $d\subA{A}=d\subA{B}=d$ can be represented by vectors in $\mathbb{R}^{d^{2}-1}$ with components $\tilde{t}_{n}$ and density operators
\vspace*{-4mm}
\begin{align}
    \rho    &=\,\frac{1}{d^{2}}\Bigl(\mathds{1}\subA{AB}\,+\,
    \sum\limits_{n=1}^{d^{2}-1}\,\tilde{t}_{n}\,\sigA_{n}\otimes\sigB_{n}\Bigr)\,.
    \vspace*{-1mm}
    \label{eq:density operator bloch decomp Weyl states}
\end{align}
The vector components $\tilde{t}_{n}$ are constrained by the positivity of $\rho$, and the allowed vectors map out a convex set in $\mathbb{R}^{d^{2}-1}$. In the case of two qubits, i.e., when $d=2$, the Weyl states can hence be nicely illustrated in $\mathbb{R}^{3}$, where $-1\leq\tilde{t}_{n}\leq+1$ and (up to local unitaries) the set $\mathcal{W}$ forms a tetrahedron, shown in Fig.~\ref{fig:tetrahedron of physical states}. The four Bell states $\ket{\Phi^{\pm}}=\tfrac{1}{\sqrt{2}}\bigl(\ket{00}\pm\ket{11}\bigr)$ and $\ket{\Psi^{\pm}}=\tfrac{1}{\sqrt{2}}\bigl(\ket{01}\pm\ket{10}\bigr)$, where $\ket{0}$ and $\ket{1}$ are the eigenstates of the third Pauli matrix with eigenvalues $+1$ and $-1$, respectively, are located at the four corners of the tetrahedron at $(\tilde{t}_{1},\tilde{t}_{2},\tilde{t}_{3})=(1,-1,1)$ ($\ket{\Phi^{+}}$), $(-1,1,1)$ ($\ket{\Phi^{-}}$), $(1,1,-1)$ ($\ket{\Psi^{+}}$), and $(-1,-1,-1)$ ($\ket{\Psi^{-}}$), while the maximally mixed state $\rho_{\mathrm{mix}}=\tfrac{1}{4}\mathds{1}\subA{4}$ is located at the origin $(0,0,0)$.\\
\vspace*{-4.5mm}

The region of separability is determined by the so-called positive partial transpose (PPT) criterion established by Peres~\cite{Peres1996} and the Horodecki family~\cite{HorodeckiMPR1996}. The criterion allows to identify bipartite quantum states as entangled, if the partial transposition of their density operator does not yield a positive operator. Given a density matrix $\rho$ in the general Bloch decomposition of Eq.~(\ref{eq:density operator bloch decomp}), the partial transposition corresponds to the transposition of the (generalized) Pauli operators in one of the subsystems, e.g., $(\sigA_{n})_{kl}\rightarrow(\sigA_{n})_{lk}$. In $2 \times 2$ and $2 \times 3$ dimensions the PPT criterion is necessary and sufficient to detect entanglement, but in higher dimensions entangled states can have a positive partial transpose. In our example of the Weyl states, the positivity constraint of the partial transpose identifies the separable Weyl states to lie within a double pyramid (see Refs.~\cite{BertlmannNarnhoferThirring2002, VollbrechtWerner2000, Horodecki-R-M1996}) with corners at $(\tilde{t}_{1},\tilde{t}_{2},\tilde{t}_{3})=(\pm1,0,0)$, $(0,\pm1,0)$ and $(0,0,\pm1)$, as illustrated in Fig.~\ref{fig:tetrahedron of physical states}. The maximal Ku{\'s}-{\.Z}yczkowski ball~\cite{kus-zyczkowski2001} of absolutely separable states lies within the

\newpage
\begin{figure}
	\centering
    %%%trim={<left> <lower> <right> <upper>}
	\includegraphics[width=0.48\textwidth,trim={0cm 0.95cm 0cm 0.6cm},clip]{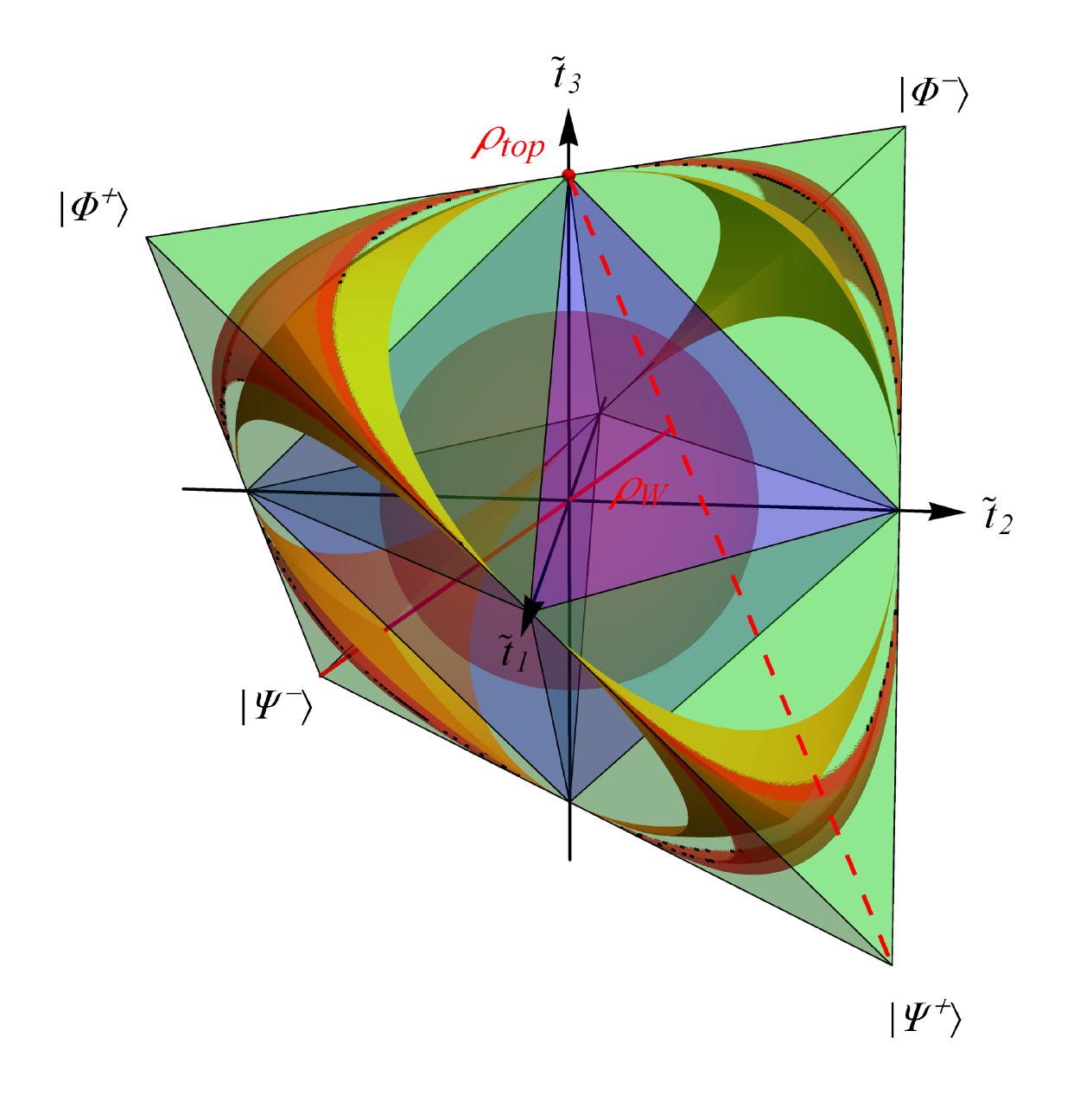}
%Tetraeder04_corr_2.pdf}
    \vspace*{-6mm}
	\caption{\textbf{Tetrahedron of Weyl states}. All states in the set $\mathcal{W}$ can (up to local unitaries) be geometrically represented as a tetrahedron spanned by the four Bell states $\ket{\Phi^{\pm}}$, $\ket{\Psi^{\pm}}$. The singular values $(\tilde{t}_{1},\tilde{t}_{2},\tilde{t}_{3})$ of the correlation tensor serve as coordinates. The set $\mathcal{S}$ of separable states forms a double pyramid (blue) and the entangled states are located in the remaining corners outside. The maximally mixed state $\rho_{\mathrm{mix}}$ is located at the origin, and the maximal ball of absolutely separable states (purple) around $\rho_{\mathrm{mix}}$ is contained within $\mathcal{S}$, but touches the double pyramid at the central points of its eight faces. The states located at the tips of the double pyramid, for instance, the Narnhofer state $\rho_{\rm{N}}$ at the coordinates $(1,0,0)$, are the separable states with maximal purity. All states that cannot violate the CHSH-Bell inequality lie within the dark-yellow, curved surfaces drawn outside the double pyramid, which includes also some entangled states, whereas all states with positive conditional entropy lie within the outermost curved surfaces (red), as discussed in Section~\ref{sec:geometry-two-qubit-states}. The solid red line indicates the Werner states $\rho_{\mathrm{W}}$ of Eq.~(\ref{eq:Werner states}), while the dashed red line represents the subset of Gisin states with $\theta=\pi/4$, see Eq.~(\ref{eq:Gisin state}).}
	\label{fig:tetrahedron of physical states}
\vspace*{-3mm}
\end{figure}
\vspace*{-3mm}
\noindent
double pyramid and touches the faces of the pyramids at the points where $|\tilde{t}_{1}|=|\tilde{t}_{2}|=|\tilde{t}_{3}|=\tfrac{1}{3}$, four of which mark the closest separable states to the four Bell states. The entangled Weyl states are located in the four corners of the tetrahedron outside the double pyramid, extending from the separable states to the maximally entangled Bell states at the tips.

\vspace*{-3mm}
\subsection{Entanglement Witnesses}\label{sec:entanglement witness}
\vspace*{-1.5mm}

The geometric picture that presents itself for the two-qubit Weyl states, i.e., the separation of separable from entangled states by planes (the faces of the double pyramid), can indeed be generalized to arbitrary dimensions. Owing to the convex structure of the set~$\mathcal{S}$ and the Hahn-Banach theorem (see, e.g., Ref.~\cite[p.~75]{ReedSimon1972}), one may define so-called entanglement witness operators via the following theorem~\cite{HorodeckiMPR1996, BertlmannNarnhoferThirring2002, Terhal2000}.

\begin{theorem}[Entanglement~Witness~Theorem]\ \\[0.8mm]
A state $\rho$ is entangled if and only if there exists a hermitian operator $W$ \textemdash\ an entanglement witness \textemdash\ such that
\vspace*{-5mm}
\begin{subequations}
\label{def-entwit}
\begin{align}
    \HSpr{W}{\rho} &<\,0\,,
    \label{eq:def-entwit ent state}\\
    \HSpr{W}{\rho_{\mathrm{sep}}} &\geq\,0\ \ \forall\,\rho_{\mathrm{sep}}\in\mathcal{S},
    \label{eq:def-entwit sep states}
\end{align}
\end{subequations}
where the Hilbert-Schmidt inner product is defined as $\HSpr{A}{B}:=\tr\bigl(A^{\dagger}B\bigr)$ for any $A,B\in\mathcal{L}(\Ha)$ and $\mathcal{S}$ denotes the set of separable states from Eq.~\emph{(\ref{set-separable-states})}.
\label{theorem:entanglement-witness-theorem}
\end{theorem}

Geometrically, a witness operator for a given state $\rho$ defines a hyperplane in the Hilbert space $\tilde{\Ha}$ that separates the set~$\mathcal{S}$ from the point representing the state $\rho$. An entanglement witness~$W_{\mathrm{opt}}$ is called \emph{optimal} if in addition to the requirements of Eq.~(\ref{def-entwit}) there exists a separable state $\rho_{\mathrm{sep}}\in \mathcal{S}$ such that $\HSpr{W_{\mathrm{opt}}}{\rho_{\mathrm{sep}}}=0$. The operator $W_{\mathrm{opt}}$ defines a tangent plane to the convex set of separable states.

On the other hand, the minimal (trace) distance of an entangled state~$\rho$ from the set~$\mathcal{S}$, the \emph{Hilbert-Schmidt measure} $ D(\rho)$ given by
\begin{align}
    D(\rho) &=\,\min_{\rho_{\mathrm{sep}}\in\mathcal{S}}\,|\hspace*{-0.5pt}|\rho-\rho_{\mathrm{sep}}|\hspace*{-0.5pt}|\,=\,
    |\hspace*{-0.5pt}|\rho-\rho_{0}|\hspace*{-0.5pt}|\,\nonumber\\
    &=\,\sqrt{\HSpr{\rho-\rho_{0}}{\rho-\rho_{0}}}
    \label{eq:Hilbert Schmidt measure}
\end{align}
can be viewed as a measure of entanglement, where the state $\rho_{0}$ is called the nearest separable state to $\rho$. An interesting connection between the Hilbert-Schmidt measure and the entanglement witness inequality arises when we define the maximal violation of the entanglement witness inequality as
\begin{align}
   B(\rho) &:=\, \max_{W} \,\bigl[\min_{\rho_{\mathrm{sep}}\in\mathcal{S}} \HSpr{W}{\rho_{\mathrm{sep}}}\,-\,
   \HSpr{W}{\rho}\bigr]\,.
   \label{max-violation-EWI}
\end{align}
Here, the minimum is taken over all separable states and the maximum over all possible entanglement witnesses $W=W^{\dagger}\in L(\Ha)$ that are suitably normalized, i.e., $|\hspace*{-0.5pt}|W|\hspace*{-0.5pt}|=1$. With this, we can formulate the following Theorem~\cite{BertlmannNarnhoferThirring2002}.

\begin{theorem}[Bertlmann-Narnhofer-Thirring]\label{theorem:BNT-theorem}\ \\[1mm]
The maximal violation of the entanglement witness inequality is equal to the Hilbert-Schmidt measure, i.e.,
\begin{align}
    B(\rho) &=\, D(\rho)\,,
    \label{B-equals-D}
\end{align}
and is achieved for $\rho \rightarrow \rho_{0}$ and $W \rightarrow W_{\mathrm{opt}}$, where the optimal entanglement witness is given by
\begin{align}
    W_{\mathrm{opt}}    &=\,\frac{\HSpr{\rho_{0}}{\rho-\rho_{0}}\,-\,(\rho-\rho_{0})}{|\hspace*{-0.5pt}|\rho-\rho_{0}|\hspace*{-0.5pt}|}\,.
    %\rho_{0} \,-\, \rho_{\rm ent} \,-\, \left\langle \rho_0|\rho_0 \,-\, \rho_{\rm ent} \right\rangle \mathds{1}}{\left\| \rho_0 \,-\, \rho_{\rm ent} \right\|} \,.
    \label{Aopt-explicit-expression}
\end{align}
\end{theorem}
\vspace*{4mm}

As an example, consider the totally antisymmetric Bell state\\
\vspace*{-6.5mm}
\begin{align}
    \rho^{-}    &=\,\ket{\Psi^{-}}\!\nl\bra{\Psi^{-}}\,=\,\tfrac{1}{4}\bigl(\mathds{1}\subtiny{0}{0}{2}\otimes\mathds{1}\subtiny{0}{0}{2}\,-\,
    \vec{\sigma}\otimes\vec{\sigma}\bigr)\,,
    \label{eq:rho-Bell-minus}
\end{align}
where $\vec{\sigma}\otimes\vec{\sigma}=\sum_{n=1}^{3}\sigma_{n}\otimes\sigma_{n}$ is used as a shorthand and the $\sigma_{n}$ are the usual Pauli matrices. The optimal entanglement witness $W^{\rho^{-}}_{\mathrm{opt}}$ for this state is given by
\begin{align}
    W_{\mathrm{opt}}^{\rho^{-}} &=\, \tfrac{1}{2\sqrt{3}}\bigl(\mathds{1}\subtiny{0}{0}{2}\otimes\mathds{1}\subtiny{0}{0}{2}\,+\,
    \vec{\sigma}\otimes\vec{\sigma}\bigr)\,.
    \label{entwit rho-Bell minus}
\end{align}
The Hilbert-Schmidt product with the entangled state of Eq.~(\ref{eq:rho-Bell-minus}) yields
\begin{align}
    \HSpr{W_{\mathrm{opt}}^{\rho^{-}}}{\rho^{-}}    &=\,\tr\bigl(\nr W_{\mathrm{opt}}^{\rho^{-}}\nr\rho^{-}\nr\bigr)\,=\,-\nr \frac{1}{\sqrt{3}}\,<0\,,
    \label{eq:rhominus ent wit HS product}
\end{align}
as required for an entanglement witness in inequality~(\ref{eq:def-entwit ent state}). To confirm that also the second inequality~(\ref{eq:def-entwit sep states}) is satisfied, first note that any separable state can be written as a convex combination of product states with local Bloch vectors $\vec{a}_{i}$ and $\vec{b}_{i}$, that is,
\begin{align}
    \rho_{\mathrm{sep}} &=\sum_{i} p_{i} \,\rho\supA{A}_{i} \otimes \rho\supA{B}_{i}\nr=\sum_{i} \frac{p_{i}}{4}\,\bigl(\mathds{1}\subA{A}+\vec{a}_{i}\vec{\sigma}\bigr)\otimes\bigl(\mathds{1}\subA{B}+\vec{b}_{i}\vec{\sigma}\bigr).
    %\nonumber%
    \label{eq:separable states}
\end{align}
The correlation tensor of any separable state hence has components $t_{mn}=\sum_{i}p_{i}a_{i}^{m}b_{i}^{n}$ and its trace is given by
\begin{align}
    \tr(t)  &=\,\sum_{i}p_{i}\,\vec{a}_{i}\cdot\vec{b}_{i}\,=\,\sum_{i}p_{i}|\vec{a}_{i}|\,|\vec{b}_{i}|\nr \cos\delta_{i}\,,
    %\nonumber
\end{align}
where $\delta_{i}$ is the angle between the Bloch vectors $\vec{a}_{i}$ and $\vec{b}_{i}$. For two qubits we have $|\vec{a}_{i}|,|\vec{b}_{i}|\leq1$ and hence $-1\leq\tr(t)\leq1$. If we then compute the Hilbert-Schmidt product of $W_{\mathrm{opt}}^{\rho^{-}}$ with an arbitrary separable state we therefore find
\begin{align}
    \HSpr{W_{\mathrm{opt}}^{\rho^{-}}}{\rho_{\mathrm{sep}}}    &=\,\tr\bigl(\nr W_{\mathrm{opt}}^{\rho^{-}}\nr\rho_{\mathrm{sep}}\nr\bigr)\,=\,
    \frac{1}{2\sqrt{3}}\bigl(\nr 1\nr +\tr(t)\nr\bigr)\,\geq\,0\,,
    %\nonumber
\end{align}
as required in~(\ref{eq:def-entwit sep states}). The nearest separable state $\rho_{0}$, for which $\HSpr{W_{\mathrm{opt}}^{\rho^{-}}}{\rho_{0}}=0$ is given by
\begin{align}
    \rho_{0}    &=\,\tfrac{1}{4}\bigl(\mathds{1}\subtiny{0}{0}{2}\otimes\mathds{1}\subtiny{0}{0}{2}\,-\,
    \tfrac{1}{3}\vec{\sigma}\otimes\vec{\sigma}\bigr)\,,
    \label{eq:rho-Bell-minus nearest sep state}
\end{align}
which can be seen to lie on the face (closest to the corner representing $\ket{\Psi^{-}}$) of the double pyramid illustrated in Fig.~\ref{fig:tetrahedron of physical states}. The optimal witness $W_{\mathrm{opt}}^{\rho^{-}}$ from Eq.~(\ref{Aopt-explicit-expression}) hence defines the plane containing this face of the pyramid. Finally, we can now easily compute the Hilbert-Schmidt measure $D(\rho^{-})$ from Eq.~(\ref{eq:Hilbert Schmidt measure}) and compare it to Eq.~(\ref{eq:rhominus ent wit HS product}), obtaining
\begin{align}
    D(\rho^{-}) &=\,|\hspace*{-0.5pt}|\rho^{-}-\rho_{0}|\hspace*{-0.5pt}|\,=\,\frac{1}{\sqrt{3}}\,=\,-\,\HSpr{W_{\mathrm{opt}}^{\rho^{-}}}{\rho^{-}}\,.
\end{align}
Since $\min_{\rho_{\mathrm{sep}}\in\mathcal{S}} \HSpr{W}{\rho_{\mathrm{sep}}}=\HSpr{W_{\mathrm{opt}}^{\rho^{-}}}{\rho_{0}}=0$ we can therefore conclude from Eq.~(\ref{max-violation-EWI}) that, indeed, $D(\rho^{-})=B(\rho^{-})$, as claimed in Theorem~\ref{theorem:BNT-theorem}.

\subsection{Bell Inequalities \& Nonlocality}\label{sec:nonlocality}

Let us now turn from the geometric aspects of entanglement to the property referred to as \emph{nonlocality}. A quantum state is said to be nonlocal if it allows for the violation of a Bell inequality~\cite{Bell1964, Bell:Speakable1987}. This terminology originates in Bell's locality hypothesis for local hidden-variable theories. In such models, the possible measurement outcomes $A$ and $B$ of two (distant) parties are determined by a hidden parameter $\lambda$. These theories are local in the sense that the values $A=A(\lambda,\vec{a})$ and $B=B(\lambda,\vec{b})$ depend on their local measurement settings $\vec{a}$ and $\vec{b}$, respectively, but not on the setting of the other party. As can be shown~\cite{Bell1964, Bell:Speakable1987}, combinations of expectation values of local hidden-variable models are constrained by Bell inequalities, which may be violated by certain (entangled) quantum states.\\

%\vspace*{-1.5mm}
To be more specific, we consider the Clauser-Horne-Shimony-Holt (CHSH) inequality~\cite{ClauserHorneShimonyHolt1969, Bell:Speakable1987} which, in analogy to the entanglement witness inequalities (Theorem~\ref{theorem:entanglement-witness-theorem}), can be written as
\begin{align}
    \HSpr{2\mathds{1}-\B\subtiny{0}{0}{\mathrm{CHSH}}}{\rho_{\nr\mathrm{loc}}} &\geq\,0\,,
    %\langle \rho_{\,\rm loc} | \, 2\cdot\mathds{1} \,-\, \B_{\rm{CHSH}} \rangle \;\geq\; 0 \;,
    \label{eq:CHSH-inequality-density-matrix-separable}
\end{align}
where the CHSH-Bell operator $\B\subtiny{0}{0}{\mathrm{CHSH}}$ is given by
\begin{align}
    \B\subtiny{0}{0}{\mathrm{CHSH}}  &=\; \vec{a}\cdot\vecsigA \otimes (\vec{b} - \vec{b}^{\nr\prime})\cdot\vecsigB \,+\,
    \vec{a}^{\nr\prime}\cdot\vecsigA \otimes (\vec{b} + \vec{b}^{\nr\prime})\cdot\vecsigB \;,
    \nonumber%\label{CHSH-Bell operator}
\end{align}
and the vectors $\vec{a},\nr\vec{b},\nr\vec{a}^{\nr\prime},\nr\vec{b}^{\nr\prime}\in\mathbb{R}^{3}$ denote the measurement directions. All local (and all separable) states $\rho_{\nr\mathrm{loc}}$ satisfy the inequality~(\ref{eq:CHSH-inequality-density-matrix-separable}). On the other hand, for nonlocal states, like the maximally entangled Bell state $\rho^{-}$ of Eq.~(\ref{eq:rho-Bell-minus}), the CHSH inequality can be violated for some choice of measurement directions. That is, there exist states $\rho_{\mathrm{nonloc}}$ and settings $\vec{a},\nr\vec{b},\nr\vec{a}^{\nr\prime},\nr\vec{b}^{\nr\prime}$ such that
\begin{align}
    \HSpr{2\mathds{1}-\B\subtiny{0}{0}{\mathrm{CHSH}}}{\rho_{\mathrm{nonloc}}} &<\,0\,,
    \label{eq:CHSH-inequality-density-matrix-Bellstate}
\end{align}
mirroring Eq.~(\ref{eq:def-entwit ent state}) in Theorem~\ref{theorem:entanglement-witness-theorem}. Since all separable states are local, the operator $(2\mathds{1}-\B\subtiny{0}{0}{\mathrm{CHSH}})$ can be seen as a witness for nonlocality, and as a (non-optimal) entanglement witness. Unfortunately, this witness is not useful for arbitrary measurement settings. However, the cumbersome task of explicitly determining the directions $\vec{a},\nr\vec{b},\nr\vec{a}^{\nr\prime},\nr\vec{b}^{\nr\prime}$ can be circumvented via another powerful theorem by the Horodecki family~\cite{HorodeckiRPM1995}.\\

%\vspace*{-2.5mm}
\begin{theorem}[CHSH operator criterion]\label{theorem:CHSH criterion}\ \\[1mm]
Let $\rho$ be the density operator of a two-qubit state with correlation tensor $t=(t_{mn})$, see Eq.~\emph{(\ref{eq:density operator bloch decomp})}, and let $\mu_{1}$ and $\mu_{2}$ be the two largest eigenvalues of $M_{\rho}=t^{T}t$. The state is nonlocal if $\overline{\B}\subtiny{0}{0}{\mathrm{CHSH}}\suptiny{0}{0}{\,\mathrm{max}}$, the maximally possible expectation value of the Bell-CHSH operator, is larger than $2$, i.e., if
\begin{align}
    \overline{\B}\subtiny{0}{0}{\mathrm{CHSH}}\suptiny{0}{0}{\,\mathrm{max}} &=\,
    \max_{\vec{a},\nr\vec{b},\nr\vec{a}^{\nr\prime},\nr\vec{b}^{\nr\prime}}\expval{\B\subtiny{0}{0}{\mathrm{CHSH}}} \,=\,
    2\sqrt{\mu_{1}+\mu_{2}}\,>\,2\,.
    \label{eq:max BellCHSH value}
\end{align}
\end{theorem}

Using the CHSH operator criterion, it is straightforward to verify that there are quantum states that are entangled, but nonetheless local in the sense of the CHSH inequality. This is best exemplified by a certain family of bipartite mixed states, the so-called \emph{Werner states}~\cite{Werner1989}, given by
\begin{align}
    \rho_{\mathrm{W}}   &=\,\alpha\,\rho^{-}\,+\,\tfrac{1}{4}(1-\alpha)\mathds{1}\subtiny{0}{0}{4}\,=\,
    \tfrac{1}{4}\bigl(\mathds{1}\subtiny{0}{0}{2}\otimes\mathds{1}\subtiny{0}{0}{2}\,-\,\alpha\,
    \vec{\sigma}\otimes\vec{\sigma}\bigr).
    \label{eq:Werner states}
\end{align}
For the parameter range $0\leq\alpha\leq1$, the state $\rho_{\mathrm{W}}(\alpha)$ can be viewed as an incoherent mixture of the maximally entangled Bell state $\ket{\Psi^{-}}$ with probability $\alpha$ on one hand, and the maximally mixed state $\tfrac{1}{4}\mathds{1}\subtiny{0}{0}{4}$ with probability $(1-\alpha)$ on the other. However, $\rho_{\mathrm{W}}(\alpha)$ represents a valid density operator also for the range $-\tfrac{1}{3}\leq\alpha\leq0$. Geometrically this can be understood as $\alpha$ parameterizing a straight line in Fig.~\ref{fig:tetrahedron of physical states}, that connects the corner representing $\ket{\Psi^{-}}$ (for $\alpha=1$) with $\rho_{\mathrm{mix}}=\tfrac{1}{4}\mathds{1}\subtiny{0}{0}{4}$ (for $\alpha=0$) at the origin, but continues onward until it intersects the opposite face of the double pyramid for $\alpha=-\tfrac{1}{3}$. As Werner discovered~\cite{Werner1989}, the state $\rho_{\mathrm{W}}(\alpha)$ is entangled for half its parameter range, that is, for $\tfrac{1}{3}<\alpha\leq1$, the partial transpose of $\rho_{\mathrm{W}}(\alpha)$ has a negative eigenvalue. However, the correlation tensor for this state is found to be $t_{\alpha}=-\alpha\mathds{1}\subtiny{0}{0}{3}$ and the CHSH operator criterion (Theorem~\ref{theorem:CHSH criterion}) hence informs us that $\overline{\B}\subtiny{0}{0}{\mathrm{CHSH}}\suptiny{0}{0}{\,\mathrm{max}}=2\sqrt{2\alpha^{2}}$. Consequently, the Werner state is nonlocal for $\alpha>\tfrac{1}{\sqrt{2}}$. This means that, in the range $\tfrac{1}{3}<\alpha\leq\tfrac{1}{\sqrt{2}}$ the states in Eq.~(\ref{eq:Werner states}) are entangled but nevertheless cannot violate the CHSH inequality.

Interestingly, there exist other Bell inequalities that are more efficient than the CHSH inequality in the sense that one may find states which violate the former, but not the latter. For instance, in Ref.~\cite{Vertesi2008}, a Bell-type inequality was introduced for which the Werner states show nonlocality already when $\alpha>0.7056$, which is slightly smaller than $\tfrac{1}{\sqrt{2}}\approx0.7071$. At the same time, recent improvement~\cite{HirschQuintinoVertesiNavascuesBrunner2016} of a known bound~\cite{AcinGisinToner2006} has revealed that Bell inequalities based on projective measurements cannot be violated by Werner states with $\alpha\leq0.682$, leaving only a small window of uncertainty. By employing general positive-operator-valued measurements (POVMs), one may in principle even go beyond the results for projective measurements and the Werner states may be nonlocal also for values of $\alpha$ below $0.682$. Bounds on the region of nonlocality have also been obtained in this case. In Ref.~\cite{HirschQuintinoVertesiNavascuesBrunner2016} it was shown that the correlations of Werner states with $\alpha<0.4547$ can be explained by local hidden-variable models for any measurement (improving on the previously known bounds $0.416$~\cite{Barrett2002} and $0.4519$~\cite{OszmaniecGueriniWittekAcin2016}).\\

In general one may in fact even find states with positive partial transposition that can violate certain Bell inequalities~\cite{VertesiBrunner2014}. The relationship of nonlocality with the PPT criterion, bound entanglement~\cite{Duer2001,AugusiakHorodecki2006,VertesiBrunner2012}, or steering criteria~\cite{MoroderGittsovichHuberGuehne2014} is hence complicated. For example, there are states whose entanglement is bound (no pure entangled state may be distilled from any number of copies of the state), which may yet violate a Bell inequality. Conversely, there are states with non-positive partial transposition (NPT) that do not violate any Bell inequality. And while all entangled states with positive partial transpose are bound entangled, it is not known whether a non-positive partial transposition implies distillability. For the remainder of this paper we will therefore focus on nonlocality in the sense of the violation of the CHSH inequality.\\

To incorporate this notion of nonlocality into our geometric picture, one can systematically apply the CHSH operator criterion to all Weyl states, noting that all locally maximally mixed states for which $\max_{i\neq j}\bigl[\tilde{t}_{i}^{\nr2}+\tilde{t}_{j}^{\nr2}\bigr]>2$ are nonlocal. The resulting region of nonlocality is illustrated in Fig.~\ref{fig:tetrahedron of physical states} where it is situated in the four corners of the tetrahedron outside the dark-yellow "parachutes". The region of local states can be found within these parachutes and contains all separable but also a number of (mixed) entangled states~\cite{ThirringBertlmannKoehlerNarnhofer2011, SpenglerHuberHiesmayr2011}.

\subsection{Hidden Nonlocality}\label{sec:hidden nonlocality}

Since, as Werner demonstrated~\cite{Werner1989}, certain entangled mixed states may satisfy all possible Bell inequalities, locality is not a sufficient criterion for separability. At this point it is important to note that the definition of nonlocality that we have used here is not the only one possible. Indeed, we call states nonlocal only if they can be directly used to violate a Bell inequality. However, as shown by Gisin~\cite{Gisin1996}, for some initially local quantum states the entanglement may be amplified by local filtering operations to allow for the violation of a Bell inequality. In this way the nonlocal character of the quantum system can be revealed (see also Ref.~\cite{Popescu1995} in this connection).

To understand this phenomenon, we consider a family of quantum states that arise as mixtures of pure (entangled) states $\rho_{\theta}=\ket{\psi_{\theta}}\!\!\bra{\psi_{\theta}}$, where
\begin{align}
    \ket{\psi_{\theta}} &=\,\sin(\theta)\,\ket{01}\,+\,\cos(\theta)\,\ket{10}\,,
    \label{eq:psi theta}
\end{align}
for $0<\theta<\tfrac{\pi}{2}$, with the mixed state given by
\begin{align}
    \rho_{\mathrm{top}}    &=\,\tfrac{1}{2}\bigl(\ket{00}\!\!\bra{00}+\ket{11}\!\!\bra{11}\bigr)\,=\,
    \tfrac{1}{4}\bigl(\mathds{1}\subtiny{0}{0}{2}\otimes\mathds{1}\subtiny{0}{0}{2}+
    \sigma_{z}\otimes\sigma_{z}\bigr).\label{eq:rho top}
\end{align}
The \emph{Gisin states}~\cite{Gisin1996} $\rho\subtiny{0}{0}{\mathrm{G}}$ are hence given by
\begin{align}
    &\rho\subtiny{0}{0}{\mathrm{G}}(\lambda,\theta)   \,=\,\lambda\,\rho_{\theta}\,+\,(1-\lambda)\,\rho_{\mathrm{top}}\,
    \label{eq:Gisin state}\\[1mm]
    &\ =\,\tfrac{1}{4}\bigl[\mathds{1}\subtiny{0}{0}{2}\otimes\mathds{1}\subtiny{0}{0}{2}
    -\lambda\cos(2\theta)\bigl(\sigma_{z}\otimes\mathds{1}\subtiny{0}{0}{2}-\mathds{1}\subtiny{0}{0}{2}\otimes\sigma_{z}\bigr)\nonumber\\[0.5mm]
    &\ \ +\lambda\sin(2\theta)\bigl(\sigma_{x}\otimes\sigma_{x}+\sigma_{y}\otimes\sigma_{y}\bigr)+(1-2\lambda)\,\sigma_{z}\otimes\sigma_{z}\bigr],
    \nonumber
\end{align}
for real probability weights $0\leq\lambda\leq1$. Note that the Gisin states are in general not locally maximally mixed, i.e., the local Bloch vectors do not vanish for the whole parameter range. Only the subset for which $\theta=\pi/4$ can be represented in the tetrahedron of Weyl states as a line connecting the state $\rho_{\mathrm{top}}$ at the upper corner of the separable double pyramid with the maximally entangled state $\ket{\Psi^{+}}$, as shown in Fig.~\ref{fig:tetrahedron of physical states}.

With the help of the PPT criterion one immediately finds that the Gisin state is entangled if and only if $1-\lambda<\lambda\nr\sin(2\theta)$. We can furthermore quantify the entanglement of $\rho\subtiny{0}{0}{\mathrm{G}}$ using an entanglement monotone called concurrence~\cite{Hill-Wootters1997,Wootters1998,Wootters2001}. For an arbitrary two-qubit density operator~$\rho$, the concurrence $C[\rho]$ is given by
\begin{align}
    C[\rho] &=\,\max\{0,\sqrt{\lambda_{1}}-\sqrt{\lambda_{2}}-\sqrt{\lambda_{3}}-\sqrt{\lambda_{4}}\}\,,
    \label{eq:concurrence}
\end{align}
where the $\lambda_{i}$ $(i=1,2,3,4)$ are the (nonnegative) eigenvalues of $\rho\,\sigma_{y}\hspace*{-1pt}\otimes\hspace*{-1pt}\sigma_{y}\rho^{*}\sigma_{y}\hspace*{-1pt}\otimes\hspace*{-1pt}\sigma_{y}$ in decreasing order ($\lambda_{1}\geq\lambda_{2}\geq\lambda_{3}\geq\lambda_{4}$), and $\rho^{*}$ is the complex conjugate of $\rho$ with respect to the computational basis. For the Gisin state, a simple calculation reveals that
\begin{align}
    C[\rho\subtiny{0}{0}{\mathrm{G}}]    &=\,\max\{0,\lambda\nr\sin(2\theta)+\lambda-1\}\,,
    \label{eq:Gisin state concurrence}
\end{align}
which is illustrated in Fig.~\ref{fig:Gisin state} for the allowed range of $\lambda$ and $\theta$. In contrast, we can determine the parameter range for which $\rho\subtiny{0}{0}{\mathrm{G}}(\lambda,\theta)$ is nonlocal using the CHSH operator criterion of Theorem~\ref{theorem:CHSH criterion}. Reading off the matrix elements of the correlation tensor from the Bloch decomposition \phantom{in}

\vspace*{-5mm}
\begin{figure}[hb!]
	\centering
    %%%trim={<left> <lower> <right> <upper>}
	\includegraphics[width=0.47\textwidth,trim={0cm 0cm 0cm 0cm},clip]{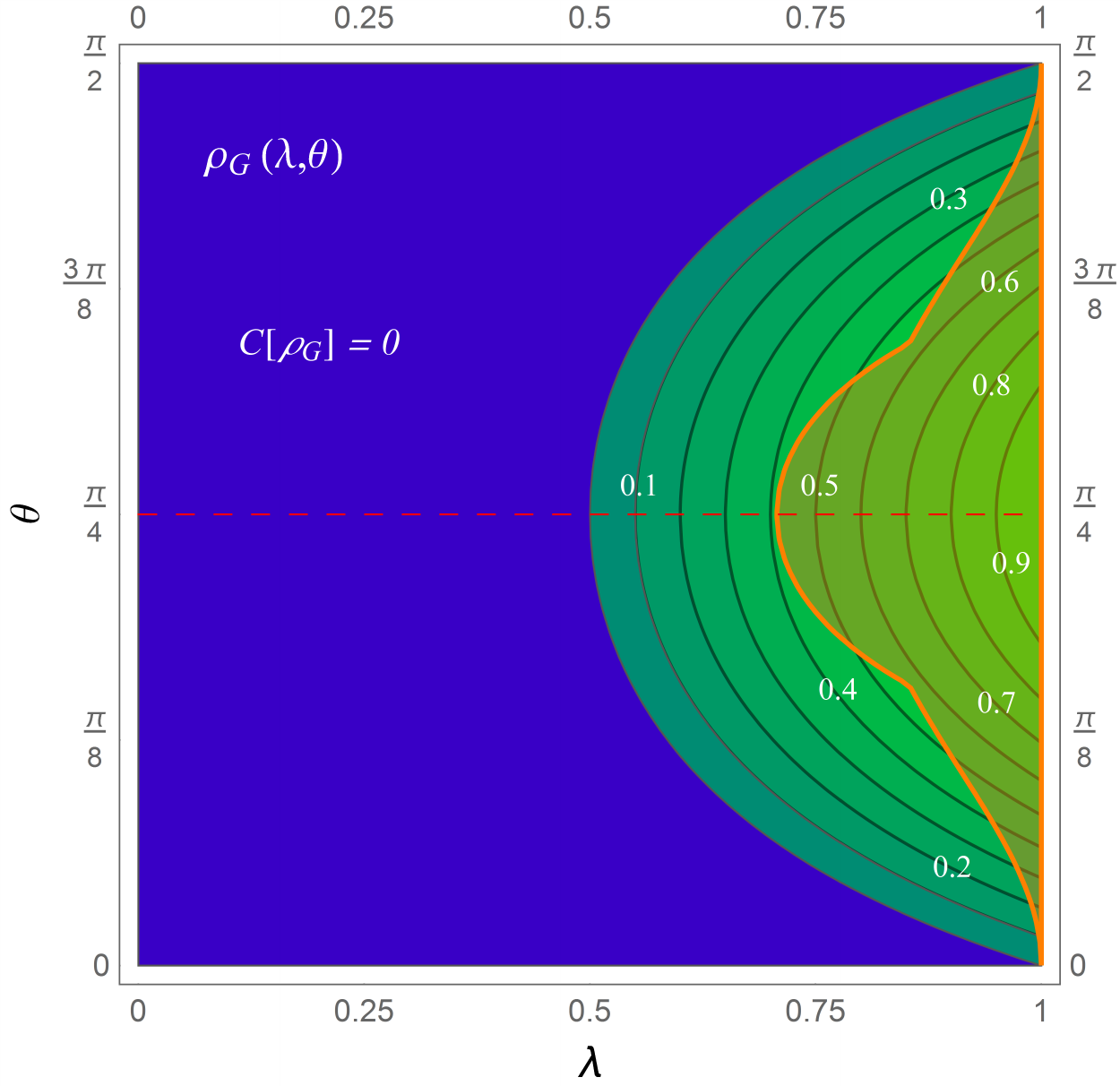}
\vspace*{-3mm}
%Tetraeder04_corr_2.pdf}
	\caption{\textbf{Gisin states}. The parameter regions of entanglement and nonlocality are shown for the family of Gisin states from Eq.~(\ref{eq:Gisin state}). Parameters $(\lambda,\theta)$ lying in the blue region on the left-hand side describe separable states. For the remaining entangled region (green) on the right-hand side the contour lines of the concurrence $C[\rho\subtiny{0}{0}{\mathrm{G}}]$ from Eq.~(\ref{eq:Gisin state concurrence}) are drawn for values $0.1$ to $0.9$ in steps of $0.1$. The orange lines on the right-hand side delimit the region of nonlocal Gisin states in the sense that $\overline{\B}\subtiny{0}{0}{\mathrm{CHSH}}\suptiny{0}{0}{\,\mathrm{max}}(\rho\subtiny{0}{0}{\mathrm{G}})>2$. The horizontal dashed red line indicates those Gisin states that are also Weyl states, corresponding to the dashed red line in Fig.~\ref{fig:tetrahedron of physical states}.}
	\label{fig:Gisin state}
\end{figure}

\clearpage
\newpage
\noindent
in Eq.~(\ref{eq:Gisin state}), one finds the maximally possible expectation value of the Bell-CHSH operator to be
\begin{align}
    \overline{\B}\subtiny{0}{0}{\mathrm{CHSH}}\suptiny{0}{0}{\,\mathrm{max}}(\rho\subtiny{0}{0}{\mathrm{G}})  &=2
    \max\Bigl\{\!\sqrt{\!\lambda^{2}\sin^{2}\!(2\theta)\!+\!(1\!-\!2\lambda)^{2}},\sqrt{2}\lambda\nr\sin(2\theta)\!\Bigr\}.
    \label{eq:Gisin state nonlocality}
\end{align}
The parameter region for which the Gisin states are nonlocal is indicated in Fig.~\ref{fig:Gisin state}. Similar to the Weyl states in Fig.~\ref{fig:tetrahedron of physical states}, some of the Gisin states may be local although being more entangled (as measured by the concurrence) than some of the nonlocal Gisin states.

However, the most interesting feature of the Gisin states is revealed by applying a local filtering procedure. That is, suppose that after sharing the state $\rho\subtiny{0}{0}{\mathrm{G}}(\lambda,\theta)$ for some $\theta$ between $0$ and $\pi/4$, Alice and Bob locally amplify their qubit states $\ket{0}\subtiny{0}{0}{A}$ and $\ket{1}\subtiny{0}{0}{B}$, respectively. Since in that case $\sin(\theta)<\cos(\theta)$, this increases the component of $\ket{01}$ with respect to that of $\ket{10}$ in $\ket{\psi_{\theta}}$, which effectively moves the state closer to the maximally entangled state $\ket{\Psi^{+}}$. Likewise, if $\pi/4<\theta<\pi/2$, amplifying $\ket{1}\subtiny{0}{0}{A}$ and $\ket{0}\subtiny{0}{0}{B}$, respectively, will have the same effect. Mathematically, these filtering operations are represented by a family of local, completely positive, and trace-nonincreasing maps $\mathcal{F}_{\theta}$, parameterized by $\theta$, and given by
\begin{align}
      \mathcal{F}_{\theta}:&\ \rho\,\mapsto\,\mathcal{F}_{\theta}(\rho)\,=\,F_{\theta}\,\rho\,F_{\theta}^{\dagger}\,.
\end{align}
Here, we choose Kraus operators satisfying $F_{\theta}^{\dagger}F_{\theta}\leq\mathds{1}$ which are given by
\begin{align}
    F_{\theta}    &=\,
    \begin{cases}
        \ F_{0}(\theta)\otimes F_{1}(\theta)  &   \mbox{if}\ \ \nr0<\theta\leq\tfrac{\pi}{4}\\
        F_{1}^{-1}(\theta)\otimes F_{0}^{-1}(\theta)  &   \mbox{if}\ \ \nl\nl\tfrac{\pi}{4}<\theta<\tfrac{\pi}{2}
    \end{cases},
\end{align}
where the local operations are
\begin{align}
    F_{0}(\theta)   &=\begin{pmatrix} \,\,1 & 0 \\[0.5mm] \,\,0 & \sqrt{\tan(\theta)} \end{pmatrix},
    \ \mbox{and}\,
    F_{1}(\theta)\,=\begin{pmatrix} \sqrt{\tan(\theta)} & 0 \\[0.5mm] 0\,\, & 1\,\, \end{pmatrix}.
    \nonumber
\end{align}
The probability for successful filtering can be computed as
\begin{align}
    \tr\bigl(\mathcal{F}_{\theta}(\rho)\bigr)   &=\,
    \begin{cases}
        \bigl(\lambda\nr\sin(2\theta)+1-\lambda\bigr)\cot(\theta)  &   \mbox{if}\ \ \nr0<\theta<\tfrac{\pi}{4}\\
        \bigl(\lambda\nr\sin(2\theta)+1-\lambda\bigr)\tan(\theta)  &   \mbox{if}\ \ \tfrac{\pi}{4}<\theta<\tfrac{\pi}{2}
    \end{cases}.
    \nonumber
\end{align}
With this, we get the normalized quantum state after the filtering procedure, i.e.,
\begin{align}
    &\rho\subtiny{0}{0}{\mathcal{F}}(\lambda,\theta)    \,=\,\frac{\mathcal{F}_{\theta}(\rho)}{\tr\bigl(\mathcal{F}_{\theta}(\rho)\bigr)}
    \,=\,\frac{\lambda\nr\sin(2\theta)\,\rho^{+}\,+\,(1-\lambda)\,\rho_{\mathrm{top}}}{1-\lambda+\lambda\nr\sin(2\theta)}\nonumber\\[1.5mm]
    &=\,\frac{1}{4}\bigl[\mathds{1}\subtiny{0}{0}{2}\otimes\mathds{1}\subtiny{0}{0}{2}
    +\frac{\lambda\nr\sin(2\theta)}{1-\lambda+\lambda\nr\sin(2\theta)}
    \bigl(\sigma_{x}\otimes\sigma_{x}+\sigma_{y}\otimes\sigma_{y}\bigr)\nonumber\\[0.5mm]
    &\ +\frac{1-\lambda-\lambda\nr\sin(2\theta)}{1-\lambda+\lambda\nr\sin(2\theta)}\,\sigma_{z}\otimes\sigma_{z}\bigr].
    \label{eq:filtered Gisin state}
\end{align}
The filtered Gisin state now fully lies within the set of Weyl states. In fact, the set of filtered Gisin states coincides with the set of unfiltered Gisin states for $\theta=\pi/4$, represented by the dashed red lines from $\rho_{\mathrm{top}}$ to $\rho^{+}=\ket{\Psi^{+}}\!\!\bra{\Psi^{+}}$ in Fig.~\ref{fig:tetrahedron of physical states} and Fig.~\ref{fig:Gisin state}.

Moreover, we can easily evaluate the concurrence of the filtered Gisin states, obtaining
\begin{align}
    C[\rho\subtiny{0}{0}{\mathcal{F}}]    &=\,\max\{0,\frac{\lambda\nr\sin(2\theta)-(1-\lambda)}{\lambda\nr\sin(2\theta)+1-\lambda}\}\,,
    \label{eq:filtered Gisin state concurrence}
\end{align}
which is illustrated in Fig.~\ref{fig:filtered Gisin state}. As for the unfiltered state, we see that $\rho\subtiny{0}{0}{\mathcal{F}}$ is entangled if (and only if) $\lambda\nr\sin(2\theta)>(1-\lambda)$. This means, the filtering is not able to entangle initially separable states. However, since the denominator satisfies $\lambda\nr\sin(2\theta)+1-\lambda\leq1$, all the already entangled states can be seen to become more entangled. Although the filtering operation is local, this is possible since the part of the initial quantum state that does not pass the filters is disregarded. If we were to complete the (trace-nonincreasing) quantum operation $\mathcal{F}_{\theta}$ to a (trace-preserving) quantum channel $\overline{\mathcal{F}}_{\theta}$ with Kraus operators $F_{\theta}$ and $\overline{F}_{\theta}=\bigl(\mathds{1}-F_{\theta}^{\dagger}F_{\theta}\bigr)\suptiny{0}{0}{1/2}$, then the entanglement would not increase.

Having noted that this amplification of the entanglement leaves separable states separable, it is now interesting to consider the effect on nonlocality. We calculate
\vspace*{-2mm}
\begin{figure}[hb!]
	\centering
    %%%trim={<left> <lower> <right> <upper>}
	\includegraphics[width=0.48\textwidth,trim={0cm 0cm 0cm 0cm},clip]{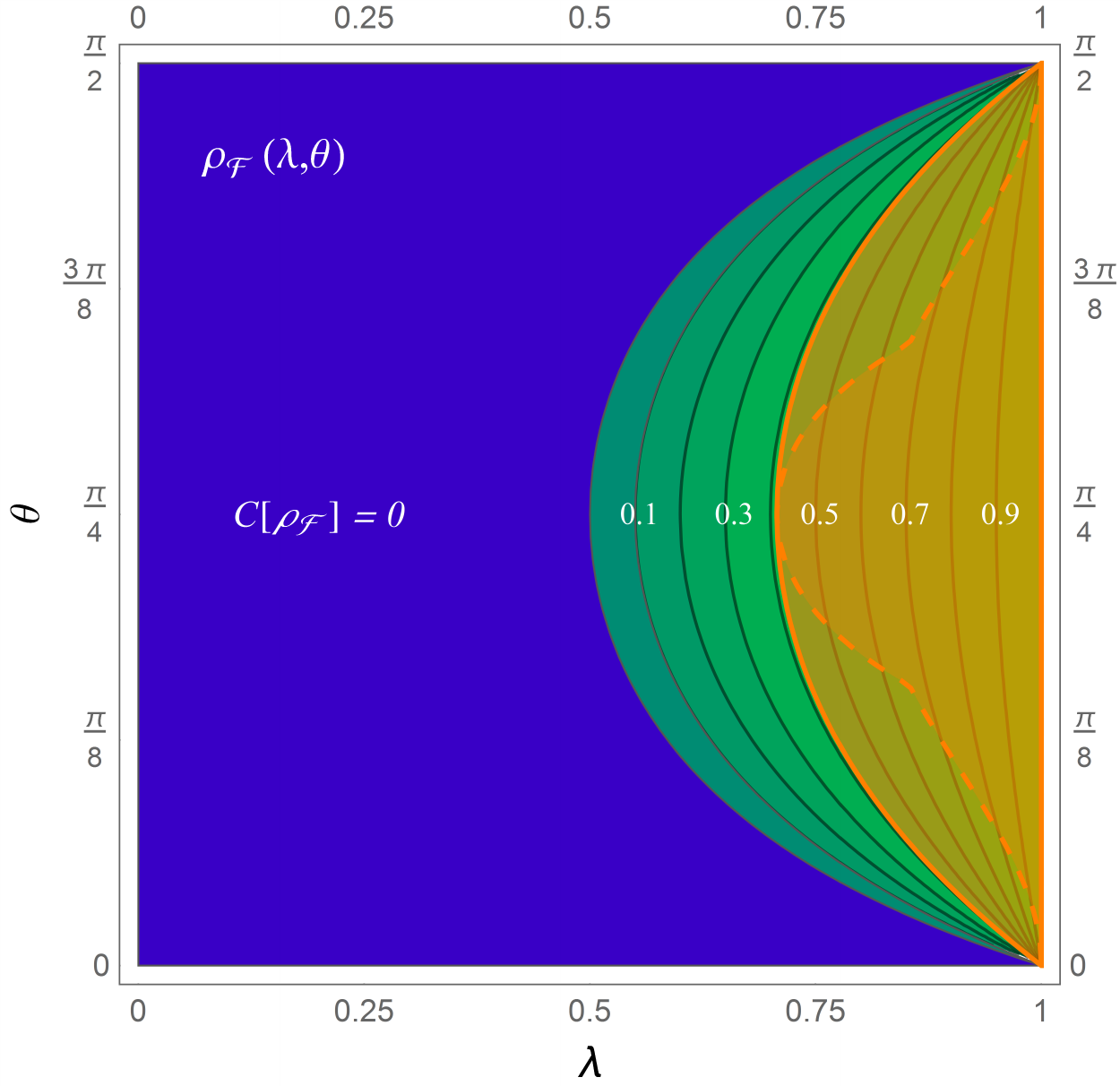}
%Tetraeder04_corr_2.pdf}
	\caption{\textbf{Filtered Gisin states}. The parameter regions of entanglement and nonlocality are shown for the family of filtered Gisin states from Eq.~(\ref{eq:filtered Gisin state}). Parameters $(\lambda,\theta)$ lying in the blue region on the left-hand side describe separable states. For the remaining entangled region (green) on the right-hand side the contour lines of the concurrence $C[\rho\subtiny{0}{0}{\mathcal{F}}]$ from Eq.~(\ref{eq:filtered Gisin state concurrence}) are drawn for values $0.1$ to $0.9$ in steps of $0.1$. The orange area on the right-hand side indicates the nonlocal filtered Gisin states in the sense that $\overline{\B}\subtiny{0}{0}{\mathrm{CHSH}}\suptiny{0}{0}{\,\mathrm{max}}(\rho\subtiny{0}{0}{\mathcal{F}})>2$. The dashed orange line delimits the parameter region for which the unfiltered Gisin states $\rho\subtiny{0}{0}{\mathrm{G}}(\lambda,\theta)$ are nonlocal, see Fig.~\ref{fig:Gisin state}.}
	\label{fig:filtered Gisin state}
\end{figure}

\clearpage
\newpage
\noindent
the maximally possible expectation value of the CHSH inequality from Theorem~\ref{theorem:CHSH criterion}, which yields
\begin{align}
    &\overline{\B}\subtiny{0}{0}{\mathrm{CHSH}}\suptiny{0}{0}{\,\mathrm{max}}(\rho\subtiny{0}{0}{\mathcal{F}})  \,=\,
    \frac{2}{\lambda\nr\sin(2\theta)+1-\lambda}\times
    \label{eq:filtered Gisin state nonlocality}\\
    &\times\max\Bigl\{\!\sqrt{\lambda^{2}\sin^{2}\!(2\theta)+\bigl(1-\lambda-\lambda\nr\sin(2\theta)\bigr)^{2}},
    \sqrt{2}\lambda\nr\sin(2\theta)\!\Bigr\}.
    \nonumber
\end{align}
Focussing on the parameter region where $\rho\subtiny{0}{0}{\mathcal{F}}$ is entangled, i.e., for $\lambda>\bigl(1+\nr\sin(2\theta)\bigl)^{-1}$, we find the condition for the filtered Gisin state to be nonlocal as
\begin{align}
    \overline{\B}\subtiny{0}{0}{\mathrm{CHSH}}\suptiny{0}{0}{\,\mathrm{max}}(\rho\subtiny{0}{0}{\mathcal{F}})\,>\,2\ \Leftrightarrow\
    \lambda\,>\,\frac{1}{(\sqrt{2}-1)\sin(2\theta)\,+\,1}
    %\lambda\nr\sin(2\theta)\,>\,\frac{1-\lambda}{\sqrt{2}-1}
    \,.
\end{align}
As illustrated in Fig.~\ref{fig:filtered Gisin state}, the nonlocal parameter region for $\rho\subtiny{0}{0}{\mathcal{F}}(\lambda,\theta)$ includes the entire region of nonlocality of the unfiltered state, but is also strictly larger. Some previously local (entangled) states become more strongly entangled and even nonlocal due to the filtering. The amplification of entanglement hence reveals the hidden nonlocality of some of the Gisin states, while others remain local. Although this separation may attributed to the choice of filtering operation, it should be remarked here that not every entangled state can become nonlocal under local filtering operations~\cite{HirschQuintinoBowlesVertesiBrunner2016}.\\

Further note that, in contrast to Gisin's \emph{nonunitary} but \emph{local} filtering operations, one may instead use a \emph{unitary} but \emph{nonlocal} operation to increase the entanglement of the Gisin state. This simply corresponds to another choice of factorizing the algebra of a density matrix~\cite{ThirringBertlmannKoehlerNarnhofer2011}. In this case, the mixedness of the state would remain unchanged. For instance, consider the unitary transformation given by
\begin{align}
    U_{\theta}  &=\,\tfrac{1}{\sqrt{2}}\bigl(f_{+}(\theta)\mathds{1}\subtiny{0}{0}{2}\otimes\mathds{1}\subtiny{0}{0}{2}
    -i\nr f_{-}(\theta)\sigma_{x}\nl\otimes\nl\sigma_{y}\bigr)\,,
\end{align}
where $f_{\pm}(\theta)=\cos(\theta)\pm\sin(\theta)$. Since $\rho_{\theta}$ from Eq.~(\ref{eq:psi theta}) is transformed to the maximally entangled state $\rho^{+}$, i.e., $U_{\theta}\rho_{\theta}U_{\theta}^{\dagger}=\rho^{+}$, and $\rho_{\mathrm{top}}=U_{\theta}\rho_{\mathrm{top}}U_{\theta}^{\dagger}$ is left invariant by the unitary transformation, the Gisin states become
\begin{align}
    \rho_{U}(\lambda)   &=\,U_{\theta}\rho_{\mathrm{G}}(\lambda,\theta)U_{\theta}^{\dagger}\,=\,
    \lambda\,\rho^{+}\,+\,(1-\lambda)\,\rho_{\mathrm{top}}\,.
    \label{eq:rotated Gisin states}
\end{align}
The unitarily transformed Gisin states are independent of $\theta$, and more specifically, $\rho_{U}(\lambda)=\rho_{\mathrm{G}}(\lambda,\pi/4)$. The unitary hence corresponds to vertically moving states in Fig.~\ref{fig:Gisin state} towards the dashed red line of $\theta=\pi/4$ while keeping $\lambda$ fixed. It can easily be seen that this allows for separable states to become entangled, and even nonlocal with respect to the new factorization.

\subsection{Classical and Quantum Entropy Measures}\label{sec:entropy measures}

Let us now turn to another major category of quantities used for the characterization of correlations. Many fundamental features of multi-party (quantum) systems can be captured by entropy, a key concept in both classical and quantum physics. In classical information theory, the basic quantity is the Shannon entropy. For a random variable $A$ whose possible  values $a$ are encountered with probability $p(a)$, the Shannon entropy is given by
\begin{align}
    H(A)    &=\,-\sum_{a}p(a)\log p(a)\,,
    \label{eq:Shannon entropy}
\end{align}
where the logarithm is understood to be to base~$2$. The Shannon entropy $H(A)$ represents the uncertainty for the occurrence of the values $a$ in the sense that it quantifies the amount of information (in bits) that is gained on average by sampling the random variable once. For a bipartite system with \emph{independent} random variables $A$ and $B$ with values $a$ and $b$, respectively, the joint probability distribution factorizes $p(a,b)=p(a)\nr p(b)$. In this case the joint entropy $H(A,B)$ is additive, i.e.,
\begin{align}
    H(A,B)  &:=\,-\sum_{a,b}\,p(a,b)\,\log\,p(a,b)\,=\,H(A)\,+\,H(B)\,,
    \nonumber
\end{align}
and any information gained about $b$ does not reveal any information about $a$, or the other way around. In general, however, the joint entropy is subadditive, that is, $H(A,B)\leq H(A)+H(B)$. The strict inequality holds when the random variables $A$ and $B$ are correlated, such that information about the occurrence of $b$ gives us information about the occurrence of $a$, and vice versa. For an arbitrary joint probability distribution $p(a,b)$, the entropy can hence be written as
\begin{align}
    H(A,B)  &=H(A|B)+H(B)=H(B|A)+H(A),
    \label{classical-joint-entropy}
\end{align}
where $H(A|B)$ is the (classical) conditional entropy defined as
\begin{align}
    H(A|B)  &=\; - \sum_{a,b} \,p(a,b) \,\log \,p(a|b) \,
    \label{eq:classical cond entropy}\\
    &=\,H(A,B)\,-\,H(B)\,,\nonumber
\end{align}
where $p(a|b) =p(a,b)/p(b)$ is the conditional probability, i.e., the probability of the occurrence of $a$ conditional on the occurrence of $b\,$. In other words, the conditional entropy of Eq.~(\ref{eq:classical cond entropy}) characterizes the uncertainty about the value $a$ when the value $b$ is already known. From the above definitions it immediately follows that
\begin{align}
    0\,\leq\, H(A|B)\,\leq\, H(A)\,.
    \label{eq:cond entropy bounds}
\end{align}
If one wishes to define a measure for the correlations between $A$ and $B$, the (classical) \emph{mutual information} $H(A\!:\!B)$ readily presents itself. It can be defined as the amount of information that is encoded in the joint distribution $p(a,b)$ but which is not contained in the local distributions $p(a)$ and $p(b)$, i.e.,
\begin{align}
    H(A\!:\!B)  &:=\,H(A)\,+\,H(B)\,-\,H(A,B)\,.
    \label{eq:classcial mutual information I}
\end{align}
Another possible definition for the mutual information is as the difference between the local uncertainty $H(A)$ and the conditional uncertainty $H(A|B)$, that is,
\begin{align}
    H(A\!:\!B)  &:=\,H(A)\,-\,H(A|B)\,.
    \label{eq:classcial mutual information J}
\end{align}
As can be easily seen from Eq.~(\ref{classical-joint-entropy}), these definitions for the mutual information are equivalent. Moreover, since $H(A,B)\leq H(A)+H(B)$, and the conditional entropies are nonnegative, $H(A|B),H(B|A)\geq0$, the two definitions in Eqs.~(\ref{eq:classcial mutual information I}) and~(\ref{eq:classcial mutual information J}) further imply the bound
\begin{align}
    0\,\leq\,H(A\!:\!B)\,\leq\,\min\{H(A),H(B)\}.
    \label{eq:bounds on classical mututal inf}
\end{align}
However, when we extend these entropic measures to quantum systems, we will encounter some interesting differences to the classical case, especially when entangled systems are considered.

The quantum analogue to the classical entropy of Eq.~(\ref{eq:Shannon entropy}) is the \emph{von~Neumann entropy} $S(\rho)$, defined as the Shannon entropy of the spectrum of the density operator $\rho$ representing the quantum state, that is,
\begin{align}
    S(\rho) &:=\,-\tr\bigl(\rho\,\log\rho\bigr)\,=\,-\sum\limits_{n}p_{n}\log p_{n}\,,
    \label{eq:von Neumann entropy}
\end{align}
where $\rho=\sum_{n}p_{n}\ket{\psi_{n}}\!\!\bra{\psi_{n}}$ for some orthonormal basis $\{\ket{\psi_{n}}\}$. Similar to the Shannon entropy, the von~Neumann entropy represents the uncertainty \textemdash\ the lack of information \textemdash\ we have about the state represented by $\rho$. This definition naturally applies to bipartite systems with density operators $\rho\subtiny{0}{0}{AB}$, such that the joint entropy is
\begin{align}
    S(A,B)  &\equiv\,S(\rho\subtiny{0}{0}{AB})\,=\,-\,\tr\bigl(\rho\subtiny{0}{0}{AB}\log\rho\subtiny{0}{0}{AB}\bigr).
\end{align}
Since the von~Neumann entropy of pure states vanishes, one may quantify the entanglement of bipartite pure states $\ket{\psi}\subtiny{-1}{0}{AB}$ by the entropy of the reduced states, i.e., one can define the \emph{entropy of entanglement} $\mathcal{E}(\ket{\psi}\subtiny{-1}{0}{AB})$ as
\begin{align}
    \mathcal{E}(\ket{\psi}\subtiny{-1}{0}{AB})  &=\,S(A)\,=\,S(B)\,,
\end{align}
where $S(A)\equiv S(\rho\subA{A})$ and $S(B)\equiv S(\rho\subA{B})$. However, when generalizing this concept to mixed states, it becomes problematic to distinguish the contributions of the joint state entropy and entanglement to the entropy of the subsystems. This necessitates the introduction of a complicated optimization procedure when defining the so-called \emph{entanglement of formation} $\mathcal{E}_{\mathrm{oF}}$ of a mixed state as
\begin{align}
    \mathcal{E}_{\mathrm{oF}}(\rho)   &=\,\min_{\{(p_{n},\ket{\psi_{n}})\}}\sum\limits_{n}p_{n}\,\mathcal{E}(\ket{\psi_{n}})\,,
\end{align}
where the minimization is carried out over all pure-state ensembles realizing the density operator $\rho=\sum_{n}p_{n}\ket{\psi_{n}}\!\!\bra{\psi_{n}}$. It is not known how to practically carry out this optimization in general, but for some special cases, $\mathcal{E}_{\mathrm{oF}}(\rho)$ can be computed explicitly. Amongst these, the most prominent is the case of two qubits, where the entanglement of formation is found to be a monotonously increasing function of the concurrence~\cite{Wootters1998} of Eq.~(\ref{eq:concurrence}), i.e.,
\begin{align}
     \mathcal{E}_{\mathrm{oF}}(\rho)    &=\,h\bigl(\frac{1+\sqrt{1-C^{2}[\rho]}}{2}\bigr)\,,
\end{align}
where $h(p)=-p\log(p)-(1-p)\log(1-p)$ is the Shannon entropy of the Bernouli distribution $\{p,1-p\}$.

In contrast, the straightforward generalization of the mutual information from Eq.~(\ref{eq:classcial mutual information I}) to the quantum case, given by
\begin{align}
    S(A\!:\!B)  &:=\,S(A)\,+\,S(B)\,-\,S(A,B)\,,
    \label{eq:quantum mutual information I}
\end{align}
does not separate genuine quantum correlations (i.e., entanglement) from purely classical correlations. Instead, as emphasized by Cerf and Adami~\cite{Cerf-Adami-PRL1997}, the quantum mutual information $S(A\!:\!B)$ is a measure of the overall correlations. Moreover, $S(A\!:\!B)$ has an interesting interpretation in the context of quantum thermodynamics. The quantum mutual information can be shown to be proportional to the work cost of its creation from an initial thermal bath~\cite{BruschiPerarnauLlobetFriisHovhannisyanHuber2015}. That is, the maximal amount of correlation as measured by the mutual information that can be created between two initially thermal, noninteracting systems at temperature $T$ at the expense of the work $W$ is $S(A\!:\!B)=W/T$ (in units where $\hbar=k\subA{\mathrm{B}}=1$).

And while the quantum mutual information also remains positive, $S(A\!:\!B)\geq0$, just as the classical mutual information in Eq.~(\ref{eq:classcial mutual information I}), it can exceed the classical upper bound from Eq.~(\ref{eq:bounds on classical mututal inf}) by a factor of $2$ such that we have
\begin{align}
    0   &\leq\,S(A\!:\!B)  \,\leq 2\,\min\{S(A),S(B)\} \;.
    \label{quantum-info-bound}
\end{align}
The quantum information bound of Eq.~(\ref{quantum-info-bound}) follows directly from the definition in Eq.~(\ref{eq:quantum mutual information I}) using the Araki-Lieb inequality
\begin{align}
    \big| S(A) \,-\, S(B)\big| \;\leq\; S(A,B) \;.
\end{align}
As we shall see in the next section, when one also introduces the generalization of the conditional entropy to the quantum regime one encounters some more surprises.

\subsection{Conditional Entropy and Conditional Amplitude Operator}\label{sec:entropy}

A straightforward generalization\footnote{Note that other generalizations for the quantum conditional entropy are possible, for which the equality in Eq.~(\ref{eq:relation of quantum mutual inf and cond entr}) does not hold~\cite{Zurek2000,OlivierZurek2001,HendersonVedral2001}.} of the conditional entropy of Eq.~(\ref{eq:classical cond entropy}) to bipartite density operators~$\rho\subA{AB}$ on a joint Hilbert space is
\begin{align}
    S(A|B)  &:=\,S(A,B)\,-\,S(B)\,.
    \label{eq:quantum cond entropy}
\end{align}
With this definition, one recovers the same relation to the mutual information as in the classical case, i.e.,
\begin{align}
    S(A\!:\!B)    &=\,S(A)\,-\,S(A|B)
    \label{eq:relation of quantum mutual inf and cond entr}
\end{align}
as in Eq.~(\ref{eq:classcial mutual information J}), and the upper bound for the conditional entropy remains as in the classical case. That is, since $S(A,B)\leq S(A)+S(B)$, one finds $S(A|B)\leq S(A)$, in analogy to the right-hand side of Eq.~(\ref{eq:cond entropy bounds}).
However, the lower bound is altered. In the quantum case one can encounter \emph{negative conditional entropy}. For instance, when we consider a pure, maximally entangled state such as $\rho^{-}=\ket{\Psi^{-}}\!\!\bra{\Psi^{-}}$, the joint entropy vanishes, $S(\rho^{-})=0$, while the local entropy is maximal, $S\bigl(\tr\subA{A}(\rho^{-})\bigr)=\log(2)$, and the conditional entropy hence is $S(A|B)=-\log(2)<0$. In general, the conditional entropy is thus bounded by the marginal entropies, i.e.,
\begin{align}
    -\,S(B) &\leq\,S(A|B)\,\leq\,S(A)\,.
\end{align}
A physical interpretation for the negative quantum conditional entropy was given in the context of state merging protocols between two observers~\cite{Horodecki-Oppenheim-Winter}. There it was found that positive values of $S(A\!:\!B)$ quantify the partial information in qubits that need to be sent from $A$ to $B$, whereas a negative conditional entropy indicates that, in addition to successfully running the protocol, a surplus of qubits remains for potential future communication. Moreover, a classical analogue of negative partial information can also be given~\cite{Oppenheim-Spekkens-Winter-0511247v2}. Other physical interpretations of negative conditional entropy arise in quantum thermodynamics~\cite{DelRioAbergRennerDahlstenVedral2011}, and when considering measurements of quantum systems, where the negative conditional entropy quantifies the amount of information in the post-selected ensembles~\cite{Salek-Schubert-Wiesner-PRA2014}. These interesting interpretations motivate considering ``entropic Bell inequalities'' whose violation implies a negative conditional entropy~\cite{Cerf-Adami-PRA1997}.

Here, we want to better understand the relationship of negative conditional entropy and entanglement. In order to do so, let us first discuss a different way to extend the classical conditional entropy of Eq.~(\ref{eq:classical cond entropy}) to the quantum case. That is, we consider the \emph{conditional amplitude operator} $\rho\subA{A|B}$ proposed by Cerf and Adami~\cite{Cerf-Adami-PRL1997, Cerf-Adami-PRA1999}, which is given by
\begin{align}
    \rho\subA{A|B}  &:=\,\exp\bigl(\log\rho\subA{AB}\,-\,\log(\mathds{1}\subA{\!A}\!\otimes\rho\subA{B})\bigr)\,,
    \label{eq:conditional amplitude operator}
\end{align}
where the exponential map is understood to be to base $2$. The conditional amplitude operator is a positive semi-definite hermitian operator defined on the support of $\rho\subA{AB}$ that takes over the role of the classical conditional probability $p(a|b)$ in the sense that one can now define the conditional entropy as
\begin{align}
    S(A|B)  &=\,-\tr\bigl(\rho\subA{AB}\log\rho\subA{A|B}\bigr)
    \label{eq:cond entropy from CAR}
\end{align}
in analogy to Eq.~(\ref{eq:classical cond entropy}). To see that this definition is equivalent to Eq.~(\ref{eq:quantum cond entropy}), simply note that

\newpage
\vspace*{-10mm}
\begin{align}
    &\tr\bigl(\rho\subA{AB}\log(\mathds{1}\subA{\!A}\!\otimes\rho\subA{B})\bigr) \,=\,
    \tr\bigl(\rho\subA{AB}\mathds{1}\subA{\!A}\!\otimes\log\rho\subA{B}\bigr)\\[1mm]
    &\ \ \ =\,\tr\subA{A}\bigl(\rho\subA{A}\mathds{1}\subA{\!A}\bigr)\,\tr\subA{B}\bigl(\rho\subA{B}\log\rho\subA{B}\bigr)\,=\,-S(B)\,.\nonumber
\end{align}
Despite this close analogy between the conditional probability distribution and the conditional amplitude operator there are some fundamental differences. Whereas $p(a|b)$ is a probability distribution satisfying $0 \leq p(a|b) \leq 1$, its quantum analogue $\rho\subA{A|B}$ is not a density matrix in general. While $\rho\subA{A|B}$ is hermitian and positive semi-definite, it can have eigenvalues larger than one, and hence $\rho_{A|B} \nleq 1\,$. Ultimately, this is what can lead to the negativity of the conditional entropy. As we have seen, a state for which this occurs is the maximally entangled Bell state. Moreover, we can immediately note that the spectrum of $\rho\subA{A|B}$ (and thus the conditional entropy) is invariant under any local unitary transformation of the form $U\subA{A}\otimes U\subA{B}\,$, which also leaves entanglement unchanged. This already suggests that the spectrum of the Cerf-Adami operator $\rho\subA{A|B}$ is related to the separability of quantum states. Indeed, the following theorem due to Cerf and Adami~\cite{Cerf-Adami-PRA1999} can be formulated.

\vspace*{-1mm}
\begin{theorem}[Cerf-Adami Theorem]\ \\[0.8mm]
The operator
\vspace*{-3mm}
\begin{align}
    \sigma\subA{AB} &:=\,-\,\log\rho\subA{A|B}
    \,=\,\log(\mathds{1}\subA{\!A}\!\otimes\rho\subA{B})\,-\,\log\rho\subA{AB}
    \label{CA sigma operator}
\end{align}
is positive semi-definite if the bipartite quantum states characterized by $\rho\subA{AB}$ are separable.
\label{theorem:Cerf-Adami-separability-theorem}
\end{theorem}
\vspace*{-1mm}

Theorem~\ref{theorem:Cerf-Adami-separability-theorem} implies that any separable bipartite state satisfies the condition $\rho\subA{A|B}\leq\mathds{1}\,$. In turn, this means that the conditional entropy is non-negative, $S(A|B) \geq 0$, for any separable state. States with negative conditional entropy must hence necessarily be entangled. Here, it is important to note that the negativity of the conditional entropy implies that (some of) the eigenvalues of $\rho\subA{A|B}$ exceed the physical boundary of unity, but the converse is not true as we demonstrate by several examples in Sect.~\ref{sec:geometry-two-qubit-states}.

Moreover, the condition $\rho\subA{A|B}\leq1$ and the positivity of $S(A|B)$ are only necessary for the separability of the quantum states, but are in general not sufficient. As realized in Ref.~\cite{Cerf-Adami-PRA1999}, there exist entangled quantum states $\rho\subA{AB}$ for which the operator $\sigma\subA{AB}$ from Eq.~(\ref{CA sigma operator}) is positive semi-definite, $\sigma\subA{AB}\geq0$, and hence $\rho\subA{A|B}\leq 1$ and $S(A|B)\geq0$. Such cases are of interest to the present work when it comes to detecting entanglement and nonlocality. The results of our investigation in $2 \times 2$ dimensions are presented in Sect.~\ref{sec:geometry-two-qubit-states}. Before we finally turn to these results, also note that an operator analogous to that of Eq.~(\ref{eq:conditional amplitude operator}) can be defined for the mutual information~\cite{Cerf-Adami-PRL1997}. The mutual amplitude operator $\rho\subA{A\!\nl:\nl\!B}$ defined as
\begin{align}
    \rho\subA{A\!\nl:\nl\!B}  &:=\,\exp\bigl(\log(\rho\subA{A}\otimes\rho\subA{B})\,-\,\log\rho\subA{AB}\bigr)\,,
    \label{mutual-amplitude-operator}
\end{align}
gives rise to the mutual information of Eq.~(\ref{eq:quantum mutual information I}) via
\vspace*{-2mm}
\begin{align}
    S(A\!\nl:\nl\!B)  &=\,-\tr\bigl(\rho\subA{AB}\log\rho\subA{A\!\nl:\nl\!B}\bigr)\,.
    \label{eq:mututal entropy from MAR}
\end{align}

\section{Results}\label{sec:results}

\subsection{Geometry of Two Qubit States with Negative Conditional Entropy}\label{sec:geometry-two-qubit-states}

We now wish to incorporate the negativity of the conditional entropy and the conditional amplitude operator bound into the geometric picture of two-qubit entanglement. To this end, we first consider again the Werner states $\rho_{\mathrm{W}}(\alpha)$ from Eq.~(\ref{eq:Werner states}). As we have previously argued, these locally maximally mixed states form a line in the tetrahedron of Weyl states, reaching from the maximally entangled state $\ket{\Psi^{-}}$ at $\alpha=1$, through the maximally mixed state at the origin for $\alpha=0$ to the opposite side of the separable double pyramid until $\alpha=-1/3$, see Fig.~\ref{fig:tetrahedron of physical states}. The Werner states are entangled for $\alpha>1/3$ and violate the CHSH inequality for $\alpha>1/\sqrt{2}$.\\

When we now compute the conditional entropy for the Werner state, we note that the eigenvalues of $\rho_{\mathrm{W}}$ are $\tfrac{1-\alpha}{4}$ (thrice degenerate) and $\tfrac{1+3\alpha}{4}$. With this we find that the boundary between negative and non-negative conditional entropy is given by the state $\rho_{\mathrm{W}}(\alpha_{0})$, where $\alpha_{0}\approx0.7476>1/\sqrt{2}$ is the solution of the transcendental equation
\begin{align}
    3(1-\alpha)\log(1-\alpha)\,+\,(1+3\alpha)\log(1+3\alpha)    &=\,4\log(2)\,.
    \nonumber
\end{align}
The condition $S(A|B)<0$ is hence a strictly stronger condition than nonlocality for the family of Werner states, as illustrated in Fig.~\ref{fig:Werner states}. Indeed, a numerical analysis shows that this is the case for all Weyl states, i.e., the curved surfaces beyond which the conditional entropy becomes negative lie strictly outside of the local region within the orange parachutes in the tetrahedron of locally maximally mixed states, see Fig.~\ref{fig:tetrahedron of physical states}.
\begin{figure}[hb!]
	\centering
    %%%trim={<left> <lower> <right> <upper>}
	\includegraphics[width=0.49\textwidth,trim={0cm 0.1cm 0cm 0.2cm},clip]{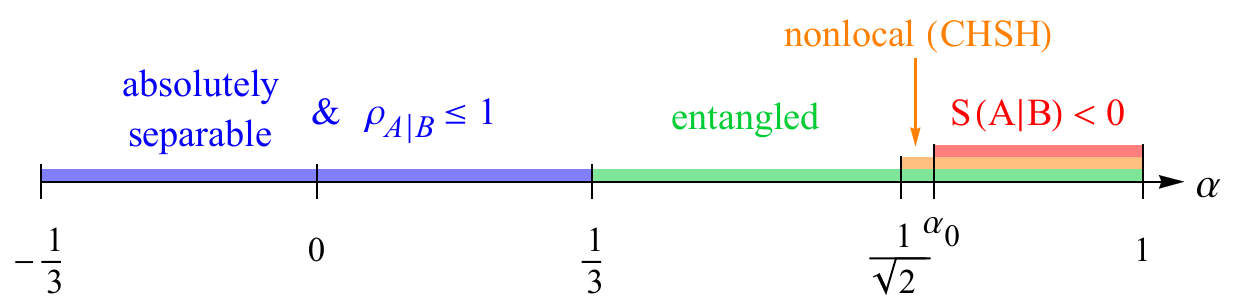}
%Tetraeder04_corr_2.pdf}
	\caption{\textbf{Werner states}. The parameter regions for the Werner state $\rho_{\mathrm{W}}(\alpha)$ from Eq.~(\ref{eq:Werner states}) are shown.}
	\label{fig:Werner states}
\end{figure}

%\clearpage
%\newpage

Examining, on the other hand, the conditional amplitude operator for the Werner states, one finds $\rho\subA{A|B}=2\rho_{\mathrm{W}}$, since $\rho\subA{B}=\tr\subA{A}(\rho_{\mathrm{W}})=\tfrac{1}{2}\mathds{1}$ such that $\log\mathds{1}\subA{A}\otimes\rho\subA{B}$ commutes with $\log(\rho_{\mathrm{W}})$. Therefore, the condition $\rho\subA{A|B}\leq\mathds{1}$ is met as long as $\alpha\leq1/3$, i.e., as long as $\rho_{\mathrm{W}}$ is separable, whereas $\rho\subA{A|B}$ has an eigenvalue larger than $1$ for $\alpha>1/3$. For the Werner states the Cerf-Adami condition $\rho\subA{A|B}\nleq\mathds{1}$ is thus equivalent to the PPT criterion~\cite{Peres1996, HorodeckiMPR1996}, a fact already noticed by Cerf and Adami~\cite{Cerf-Adami-PRA1999}.\\

Indeed, the observations we have made for the Werner states also hold for all other Weyl states in addition to this one-parameter subfamily. That is, a numerical evaluation of the conditional entropy of the locally maximally mixed states of Eq.~(\ref{eq:density operator bloch decomp Weyl states}) presented in Fig.~\ref{fig:tetrahedron of physical states} shows that the negativity of $S(A|B)$ is a strictly stronger condition than nonlocality for these states. That is, the red, curved surfaces indicating where the conditional entropy changes sign lie outside the orange parachute surfaces marking the boundary of nonlocality for all Weyl states. Moreover, we can also formulate the following conditional amplitude operator (CAO) criterion.\\

\begin{theorem}[CAO Criterion]\ \\[0.8mm]
For every locally maximally mixed state $\rho\in\mathcal{W}$, the criterion $\rho\subA{A|B}\leq\mathds{1}$ for the Cerf-Adami conditional amplitude operator $\rho\subA{A|B}$ given by Eq.~\emph{(\ref{eq:conditional amplitude operator})} is equivalent to the PPT criterion, i.e.,
\vspace*{-1.5mm}
\begin{align}
    \rho\subA{A|B}    &\leq\,\mathds{1}\qquad\mbox{if and only if $\rho\in\mathcal{W}$ is separable},\nonumber\\
    \rho\subA{A|B}    &\nleq\,\mathds{1}\qquad\mbox{if and only if $\rho\in\mathcal{W}$ is entangled}.\nonumber
\end{align}
\vspace*{-5mm}
\label{theorem:CA-operator theorem}
\end{theorem}

\begin{proof}
For the proof of Theorem~\ref{theorem:CA-operator theorem} we recall Wootters' concurrence~\cite{Hill-Wootters1997, Wootters1998, Wootters2001} from Eq.~(\ref{eq:concurrence}). For calculating $C$ we need the "spin-flipped" state $\tilde{\rho}=\sigma_{y}\hspace*{-1pt}\otimes\hspace*{-1pt}\sigma_{y}\rho^{*}\sigma_{y}\hspace*{-1pt}\otimes\hspace*{-1pt}\sigma_{y}$ which is equal to the density operator $\rho$ for all Weyl states, $\tilde{\rho}=\rho$. The square roots of the eigenvalues of $\rho\tilde{\rho}$ needed for the concurrence are hence just the eigenvalues $p_{n}$ $(n=1,2,3,4)$ of $\rho$, which satisfy $\sum_{n}p_{n}=1$. Consequently, the concurrence of all Weyl states can be written as
\begin{align}
    C[\rho]     &=\max\{0, p_{1}-p_{2}-p_{3}-p_{4}\}=\max\{0, 2\nr p_{1}-1\},
    \label{eq:concurrence formula Weyl states}
\end{align}
where the largest eigenvalue $p_{1}$ must exceed the value of $1/2$ for $\rho$ to be entangled.

Next, recall that for all Weyl states we have $\rho\subA{A}=\rho\subA{B}=\tfrac{1}{2}\mathds{1}$ and the Cerf-Adami conditional amplitude operator is hence given by
\begin{align}
    \rho\subA{A|B}  &=\,\exp\bigl(\log\rho\,-\,\log(\mathds{1}\subA{\!A}\!\otimes\rho\subA{B})\bigr)\,=\,2\rho\,.
\end{align}
Consequently, we have $\rho\subA{A|B}\nleq\mathds{1}$ when the largest eigenvalue of $\rho\subA{A|B}=2\rho$ exceeds $1$, i.e., when the largest eigenvalue of $\rho$ exceeds $1/2$. By virtue of Eq.~(\ref{eq:concurrence formula Weyl states}) this means that the state is entangled. Conversely, all entangled Weyl states must have an eigenvalue larger than $1/2$ such that $\rho\subA{A|B}\nleq\mathds{1}$. The fact that all entangled two-qubit states have nonzero concurrence and non-positive partial transposition concludes the proof.
\end{proof}

\subsection{Inequivalence of the CAO and PPT Criteria}\label{sec:counterexample}

Having established the significance of the conditional amplitude operator and the relationship of entanglement, negative conditional entropy, and nonlocality for the Weyl states, we are curious whether the observations we have made also hold for other states. We therefore consider the unitary orbit of one of the Weyl states that takes us outside this set. Starting from the Narnhofer state $\rho\subtiny{0}{0}{\mathrm{N}}=\tfrac{1}{4}(\mathds{1}\subtiny{0}{0}{2}\otimes\mathds{1}\subtiny{0}{0}{2}+\sigma_{x}\otimes\sigma_{x})$, situated at the corner of the double pyramid of separable states half-way on the line connecting $\ket{\Psi^{+}}$ and $\ket{\Phi^{+}}$ in the tetrahedron of Fig.~\ref{fig:tetrahedron of physical states}, we apply the unitary transformation
\begin{align}
    V   &=\,\tfrac{1}{4}\bigl[(2+\sqrt{2})\,\mathds{1}\subtiny{0}{0}{2}\otimes\mathds{1}\subtiny{0}{0}{2}
        +i\sqrt{2}\,(\sigma_{x}\otimes\sigma_{y}+\sigma_{y}\otimes\sigma_{x})\nonumber\\[1mm]
        &\ \ -\,(2-\sqrt{2})\,\sigma_{z}\otimes\sigma_{z}\bigr]\,.
\end{align}
The resulting state, given by
\begin{align}
    \rho_{V}    \,=\,V\rho\subtiny{0}{0}{\mathrm{N}}V^{\dagger} &=\,\tfrac{1}{4}\bigl[\mathds{1}\subtiny{0}{0}{2}\otimes\mathds{1}\subtiny{0}{0}{2}
    +\tfrac{1}{2}(\sigma_{z}\otimes\mathds{1}\subtiny{0}{0}{2}+\mathds{1}\subtiny{0}{0}{2}\otimes\sigma_{z})\nonumber\\[1mm]
    &\ +\tfrac{1}{2}(\sigma_{x}\otimes\sigma_{x}+\sigma_{y}\otimes\sigma_{y})\bigr]\,,
\end{align}
lies outside of the set $\mathcal{W}$ due to the occurrence of the term $\tfrac{1}{2}(\sigma_{z}\otimes\mathds{1}\subtiny{0}{0}{2}+\mathds{1}\subtiny{0}{0}{2}\otimes\sigma_{z})$. The purity $\tr(\rho\subtiny{0}{0}{\mathrm{N}}^{2})=\tr(\rho_{V}^{2})=\tfrac{1}{2}$ of the state is left unchanged by the unitary transformation but the final state $\rho_{V}$ is entangled. In fact, the concurrence takes the maximally possible value at this fixed purity, $C[\rho_{V}]=\tfrac{1}{2}$, i.e., the state $\rho_{V}$ belongs to the class of maximally entangled mixed states (MEMS)~\cite{IshizakaHiroshima2000,Munro-James-White-Kwiat2001}. In other words, no global unitary may entangle this state any further.

With this in mind, we now consider a family of states in the two-qubit Hilbert space along the line from $\rho_{V}$ to $\rho_{\mathrm{top}}$ from Eq.~(\ref{eq:rho top}), i.e., we define
\begin{align}
    \rho_{V}(\nu)   &:=\,\nu\,\rho_{V}\,+\,(1-\nu)\,\rho_{\mathrm{top}}\,,
\end{align}
where $0\leq\nu\leq1$. The eigenvalues of the partial transpose of $\rho_{V}(\nu)$ are $\tfrac{\nu}{4}$ (twice degenerate) and $\tfrac{1}{4}\bigl(2-\nu(1\pm\sqrt{2})\bigr)$. The states along the line are hence entangled if $\nu>2(\sqrt{2}-1)$. Now, if we consider the CAO criterion, we first compute the reduced state
\begin{align}
    \rho_{V,B}(\nu) &=\,\tr\subA{A}\bigl(\rho_{V}(\nu)\bigr)\,=\,\tfrac{1}{4}\begin{pmatrix} 2+\nu & 0 \\ 0 & 2-\nu \end{pmatrix}\,,
    \label{eq:reduced states of narnohofer-top line}
\end{align}
and the spectrum of $\rho_{V}(\nu)$, given by
\begin{align}
    \operatorname{spectr}\bigl(\rho_{V}(\nu)\bigr)   &=\,\{\tfrac{1}{2},0,\tfrac{1-\nu}{2},\tfrac{\nu}{2}\}\,.
    \label{eq:spectrum of narnohofer-top line}
\end{align}
To compute the spectrum of $\rho\subA{A|B}$, note that $\rho_{V}(\nu)$ has (at least) one vanishing eigenvalue [see Eq.~(\ref{eq:spectrum of narnohofer-top line})], which is problematic when evaluating $\log\rho_{V}(\nu)$. However, a simple work-around is to replace the vanishing eigenvalue by $\epsilon>0$ throughout the computation and take the limit $\epsilon\rightarrow0$ at the end. With this procedure we obtain the eigenvalues of $\rho\subA{A|B}$ as
\begin{align}
    \{0,\tfrac{2}{2+\nu},\tfrac{2-2\nu}{2-\nu},\tfrac{2\nu}{\sqrt{4-\nu^{2}}}\}\,.
    \label{eq:spectrum of narnohofer-top line CAO}
\end{align}
The first three eigenvalues are always smaller than $1$, but the last eigenvalue becomes larger than one when $\nu>\tfrac{2}{\sqrt{5}}>2(\sqrt{2}-1)$. We thus see that the PPT criterion and the CAO criterion are inequivalent in general. Nonetheless, the conditional entropy of the state $\rho_{V}(\nu)$ remains nonnegative for all values $\nu$, and none of these states allows for a violation of the CHSH inequality either.\\

To incorporate also negative conditional entropy and nonlocality into the picture, we hence turn again to the Gisin states $\rho\subtiny{0}{0}{\mathrm{G}}(\lambda,\theta)$ from Eq.~(\ref{eq:Gisin state}). The spectrum of the density operator $\rho\subtiny{0}{0}{\mathrm{G}}(\lambda,\theta)$ is given by
\begin{align}
    \operatorname{spectr}(\rho\subtiny{0}{0}{\mathrm{G}})   &=\,\{0,\tfrac{1-\lambda}{2},\tfrac{1-\lambda}{2},\lambda\}\,,
    \label{eq:spectrum of gisin state}
\end{align}
while the reduced states $\rho\subA{A}$ and $\rho\subA{B}$ are already diagonal and have eigenvalues $\tfrac{1}{2}\bigl(1\pm\lambda\cos(2\theta)\bigr)$. The graphical analysis of the parameter region for $\lambda$ and $\theta$ for which the conditional entropy is negative reveals an interesting feature. As can be seen in Fig.~\ref{fig:cond entr of Gisin state}~(a), while some Gisin states are both nonlocal and have negative conditional entropy, some only have one of these properties, but not the other. That is, contrary to what was found for the Weyl states, in general not all states for which $S(A|B)<0$ are also nonlocal. And, as before, not all nonlocal states have negative conditional entropy.

Following up on this surprise, let us quickly examine the condition $\rho\subA{A|B}\leq\mathds{1}$ for the Gisin states. As noted in Ref.~\cite{Cerf-Adami-PRA1999}, there exist entangled states for which $\rho\subA{A|B}\leq\mathds{1}$ indeed holds, but due to Theorem~\ref{theorem:CA-operator theorem}, these must lie outside the set $\mathcal{W}$. The Gisin states are hence perfect examples for such states.

When computing the spectrum of $\rho\subA{A|B}$ for $\rho\subtiny{0}{0}{\mathrm{G}}$ we again encounter (at least) one vanishing eigenvalue [see Eq.~(\ref{eq:spectrum of gisin state})]. As before, we therefore replace the vanishing eigenvalue by $\epsilon>0$ in the computation and consider the limit $\epsilon\rightarrow0$ at the end. With this method, the eigenvalues $\kappa_{i}$ of $\rho\subA{A|B}$ are found to be
\begin{align}
    \kappa_{1} &=\,0,\ \ \kappa_{2}\,=\,\frac{1-\lambda}{1-\lambda\nr\cos(2\theta)},\ \ \kappa_{3}\,=\,\frac{1-\lambda}{1+\lambda\nr\cos(2\theta)}\,,
    \nonumber\\[1mm]
    \mbox{and}  &\ \ \ \kappa_{4}\,=\,\frac{2\lambda\,\exp\bigl(-\cos(2\theta)\nr\artanh[\lambda\nr\cos(2\theta)]\bigr)}{\sqrt{1-\lambda^{2}\cos^{2}\!(2\theta)}}\,.
\end{align}
While $\kappa_{2}$ and $\kappa_{3}$ are smaller than $1$ for all values of $\lambda$ and $\theta$, the fourth eigenvalue $\kappa_{4}$ can become larger than $1$. The corresponding region, delimited by the purple lines in Fig.~\ref{fig:cond entr of Gisin state}~(a), is contained within the region of entangled states, but there is a region of entanglement where $\kappa_{4}<1$ and hence $\rho\subA{A|B}\leq\mathds{1}$. This clearly demonstrates that the condition $\rho\subA{A|B}\leq\mathds{1}$ for the Cerf-Adami operator in general provides a necessary but not sufficient condition for separability.

\begin{figure*}[ht!]
	\centering
    %%%trim={<left> <lower> <right> <upper>}
	(a)\includegraphics[width=0.44\textwidth,trim={0cm 0.05cm 0cm 0cm},clip]{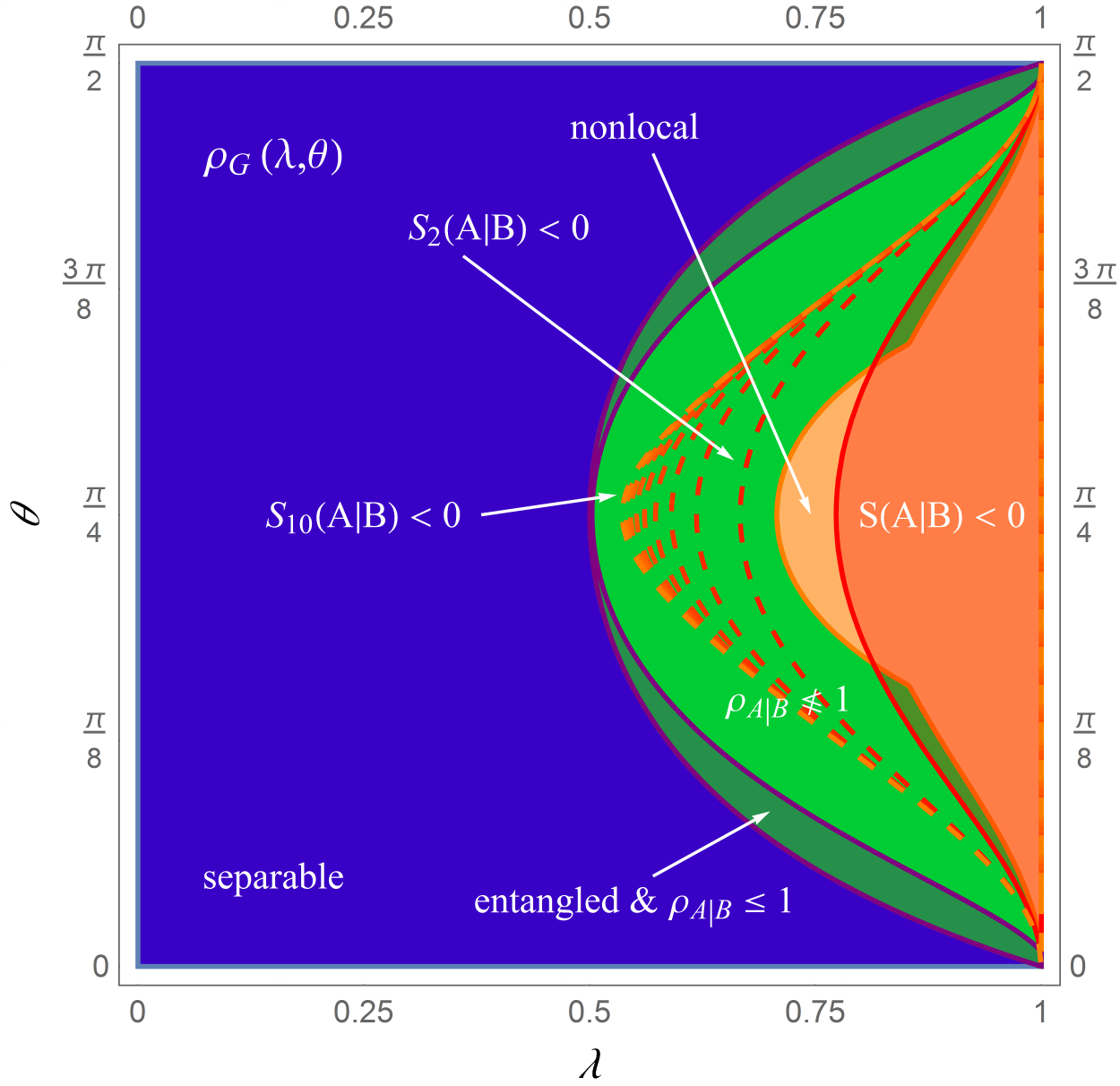}
    (b)\includegraphics[width=0.44\textwidth,trim={0cm 0.05cm 0cm 0cm},clip]{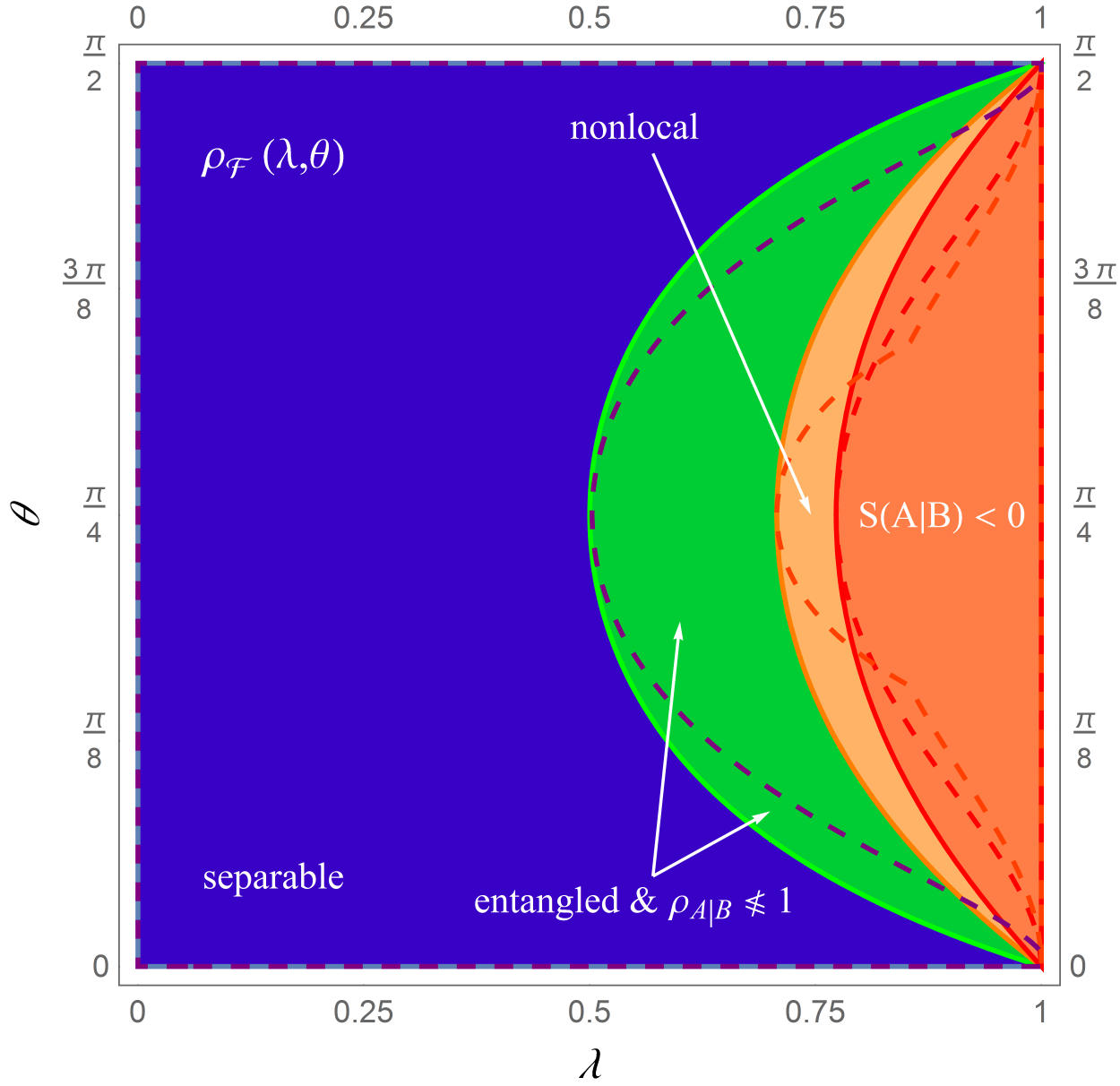}
    \vspace*{-1mm}
	\caption{\textbf{Conditional Entropy of Gisin states}.
    The parameter regions of separability (blue), entanglement (green), nonlocality (orange) and negative conditional entropy (red) are shown for the family of unfiltered Gisin states $\rho\subtiny{0}{0}{\mathrm{G}}(\lambda,\theta)$ from Eq.~(\ref{eq:Gisin state}) in (a), and for the filtered Gisin states $\rho\subtiny{0}{0}{\mathcal{F}}(\lambda,\theta)$ from Eq.~(\ref{eq:filtered Gisin state}) in (b). For the unfiltered states in (a), the region of entangled states whose conditional amplitude operator is bounded by unity, $\rho\subA{A|B}\leq\mathds{1}$, is delimited in purple, showing that the condition $\rho\subA{A|B}\nleq\mathds{1}$ is strictly weaker than the PPT criterion for two qubits. In addition, it can be clearly seen in (a) that there is no clear hierarchy between the conditions of nonlocality and negative conditional (von Neumann) entropy. That is, there exist local states with negative conditional entropy, as well as nonlocal states with positive conditional entropy, $S(A|B)\geq0$. However, the hierarchy is recovered when considering conditional R{\'e}nyi entropies $S_{\alpha}(A|B)$ from Eq.~(\ref{eq:conditional Renyi entropy}) for $\alpha\geq2$. The corresponding boundaries for $\alpha=2,3,\ldots,10$ are shown as dashed lines. In (b), the boundaries for $S(A|B)<0$, nonlocality, and $\rho\subA{A|B}\nleq\mathds{1}$ for the unfiltered states are indicated by the respective dashed lines.}
	\label{fig:cond entr of Gisin state}
\end{figure*}

For the sake of completeness and illustration, let us also re-examine the filtered Gisin states $\rho\subtiny{0}{0}{\mathcal{F}}(\lambda,\theta)$ from Eq.~(\ref{eq:filtered Gisin state}). Since these are Weyl states, Theorem~\ref{theorem:CA-operator theorem} applies and the boundary between $\rho\subA{A|B}\leq\mathds{1}$ and $\rho\subA{A|B}\nleq\mathds{1}$ coincides with the boundary between separability and entanglement. To determine the conditional entropy of $\rho\subtiny{0}{0}{\mathcal{F}}$ we note that the nonzero eigenvalues of the filtered Gisin states are
\begin{align}
    \frac{\lambda\nr\sin(2\theta)}{1-\lambda+\lambda\nr\sin(2\theta)}\ \ \ \mbox{and}\ \ \
    \frac{1-\lambda}{2\bigl(1-\lambda+\lambda\nr\sin(2\theta)\bigr)},
\end{align}
where the latter eigenvalue is twice degenerate. With these eigenvalues, we can evaluate the conditional entropy and find that the region where it is negative is contained within the region of nonlocality, see Fig.~\ref{fig:cond entr of Gisin state}~(b).

\subsection{Negativity of Generalized Conditional Entropies}\label{sec:Negativity of Generalized Conditional Entropies}

For the conditional entropy based on the von Neumann entropy $S(\rho)$, no clear hierarchy with nonlocality can hence be established in general. Some states may be nonlocal and satisfy $S(A|B)\geq0$, while other states may have negative values of $S(A|B)$, whilst being local (in the sense of the CHSH inequality). An interesting way out of this confusion is employing generalized entropy measures. One candidate for such an extension is the R{\'e}nyi $\alpha$-entropy, defined as
\begin{align}
    S_{\alpha}(\rho)    &:=\,\frac{1}{1-\alpha}\log\tr(\rho^{\alpha})\,=\,\frac{1}{1-\alpha}\log\sum\limits_{n}p_{n}^{\alpha}\,,
    \label{eq:Renyi entropy}
\end{align}
where $p_{n}$ are the eigenvalues of $\rho$ and $\alpha\geq1$. When $\alpha$ tends to $1$, the von Neumann entropy $S(\rho)$ from Eq.~(\ref{eq:von Neumann entropy}) arises from the R{\'e}nyi entropy as a limiting case, $\lim_{\alpha\rightarrow1}S_{\alpha}(\rho)=S_{1}(\rho)=S(\rho)$. For $\alpha\geq2$, this family of entropies provide stronger entanglement criteria than the von Neumann entropy, i.e., when we define the conditional R{\'e}nyi entropy $S_{\alpha}(A|B)$ as
\begin{align}
    S_{\alpha}(A|B)    &:=\,S_{\alpha}(A,B)\,-\,S_{\alpha}(B)\,,
    \label{eq:conditional Renyi entropy}
\end{align}
where $S_{\alpha}(A,B)\equiv S_{\alpha}(\rho\subA{AB})$ and $S_{\alpha}(B)\equiv S_{\alpha}(\rho\subA{B})$. These generalized conditional entropies can be shown~\cite{HorodeckiRPM1996} to be nonnegative for all separable states, such that $S_{\alpha}(A|B)<0$ implies that the quantum state is entangled. Moreover, it was shown in Ref.~\cite{HorodeckiRPM1996} that the negativity of the conditional R{\'e}nyi $2$-entropy, $S_{2}(A|B)<0$, already is a necessary condition for nonlocality in the CHSH sense. In other words, the positivity of $S_{2}(A|B)$ means that the CHSH inequality cannot be violated, i.e.,
\vspace*{-1.5mm}
\begin{align}
    S_{2}(A|B)\,\geq\,0
    \ \ \ \Rightarrow\ \ \
    \overline{\B}\subtiny{0}{0}{\mathrm{CHSH}}\suptiny{0}{0}{\,\mathrm{max}}\,\leq\,2\sqrt{1-\bigl|\,|\vec{a}|^{2}-|\vec{b}|^{2}\bigr|}\,,
    \nonumber
\end{align}
where $\vec{a}$ and $\vec{b}$ are the Bloch vectors [see Eq.~(\ref{eq:density operator bloch decomp})] of the two qubits, respectively. The conditional R{\'e}nyi $2$-entropy hence provides a strictly stronger condition than nonlocality for two qubits, but not in higher dimensions~\cite{HorodeckiRPM1996}. This is illustrated for the Weyl states in one sector of the tetrahedron in Fig.~\ref{fig:Alpha entropies tetrahedron}. Moreover, it was proven in Ref.~\cite{Horodecki-R-M1996} that all Weyl states are separable if and only if all conditional R{\'e}nyi $\alpha$-entropies are positive semi-definite, i.e.,
\begin{align}
    \rho\in\mathcal{W}\ \ \mbox{separable}\ \ \Leftrightarrow\ \ S_{\alpha}(A|B)\,\geq\,0\ \, \forall\, \alpha\,.
\end{align}
The positivity of the entire family of conditional R{\'e}nyi $\alpha$-entropies hence provides an entanglement criterion equivalent to the PPT criterion for locally maximally mixed states. For other two-qubit states things are again less clear. For instance, for the unfiltered Gisin states $\rho\subtiny{0}{0}{\mathrm{G}}(\lambda,\theta)$ from Eq.~(\ref{eq:Gisin state}), the conditions $S_{\alpha}(A|B)\geq0$ are shown in Fig.~\ref{fig:cond entr of Gisin state}~(a) for $\alpha=2,3,\ldots,10$, which are clearly stronger than nonlocality, but weaker than PPT or $\rho\subA{A|B}\leq\mathds{1}$ for detecting entanglement.

Nonetheless, conditional entropies and the conditional amplitude operator provide straightforward entanglement witnesses that can be in principle employed in systems of arbitrary dimension. For example, for some specific two-mode Gaussian states (where the Hilbert space is infinite-dimensional), the Tsallis $q$-conditional entropy gives comparable results~\cite{SudhaUshaDeviRajagopal2010} to the two-qubit case. In general, the exact relationship between entanglement, nonlocality, and conditional entropies is nonetheless complicated.

\begin{figure}
	\centering
    %%%trim={<left> <lower> <right> <upper>}
	\includegraphics[width=0.48\textwidth,trim={0cm 0.5cm 0cm 0cm},clip]{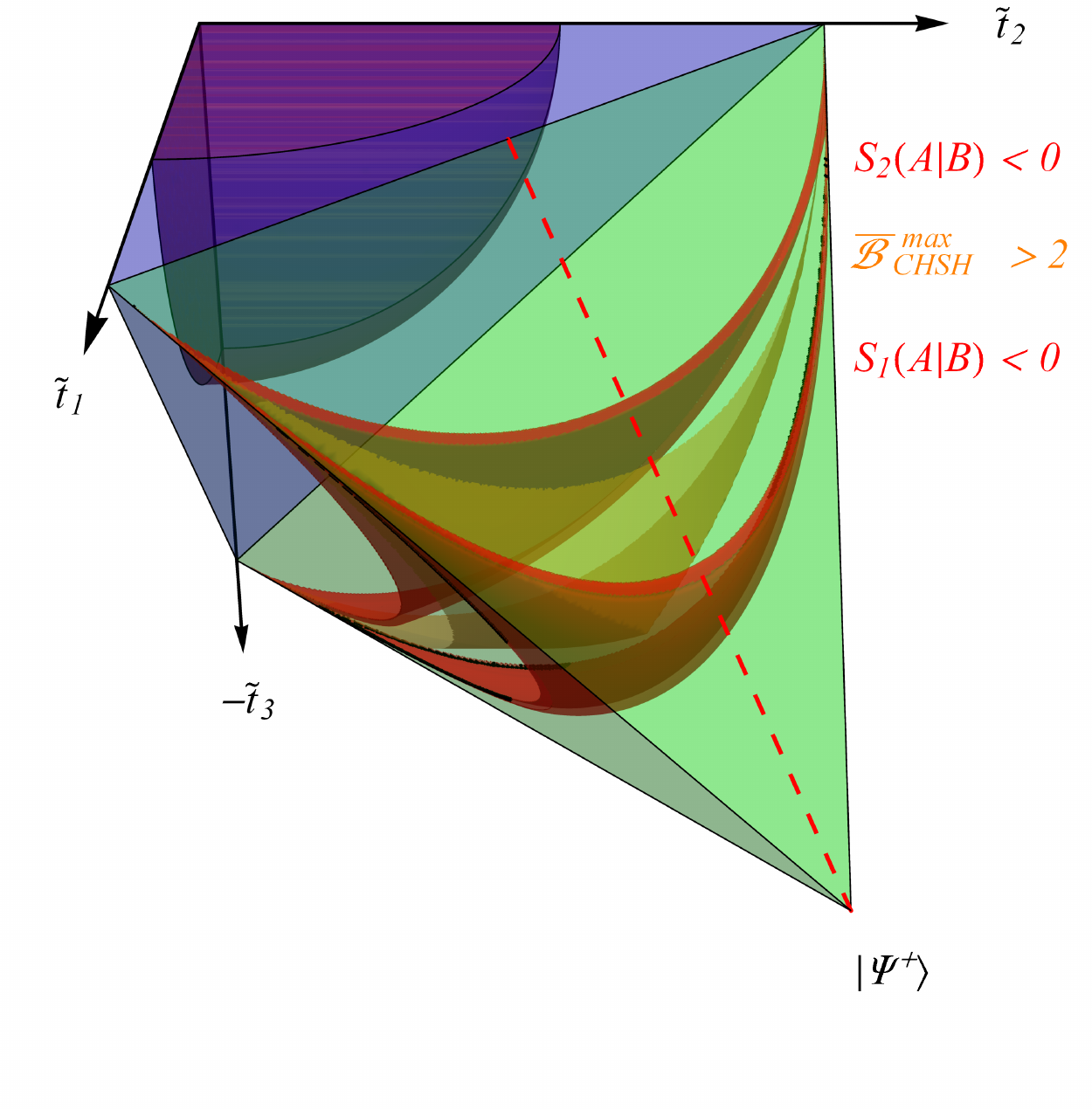}
	\caption{\textbf{Conditional R{\'e}nyi Entropies \& Nonlocality}. The sector of the tetrahedron of Weyl states defined by $\tilde{t}_{1},\tilde{t}_{2}\geq0$ and $\tilde{t}_{3}\leq0$ is shown. In addition to the boundaries shown in Fig.~\ref{fig:tetrahedron of physical states}, the boundary between states with positive and negative conditional R{\'e}nyi $2$-entropy from Eq.~(\ref{eq:conditional Renyi entropy}) is illustrated. As can be seen, all local states lie within the set of states for which $S_{1}(A|B)\geq0$, whereas all nonlocal states have negative conditional R{\'e}nyi $2$-entropy. The dashed red line indicates the filtered Gisin states in this sector of the tetrahedron, corresponding to (parts of) the dashed red lines in Fig.~\ref{fig:tetrahedron of physical states} and Fig.~\ref{fig:Gisin state}.}
	\label{fig:Alpha entropies tetrahedron}
\end{figure}

\section{Conclusion}\label{sec:conclusion}

We have reviewed the geometry of entanglement for two-qubit systems. Despite its simplicity, this bipartite system already reveals many of the intricacies in the relationship of the numerous criteria for entanglement and separability, and is hence an important guiding example. In particular, we have focussed on highlighting the roles of negative conditional entropy and the conditional amplitude operator criterion as entanglement detection methods. Since many technical complications already arise for the simple two-qubit case, we have placed specific emphasis on the family of locally maximally mixed Weyl states. For the latter, a clear hierarchy emerges, in which the set of CHSH-nonlocal states fully contains the set of states with negative conditional (von Neumann) entropy, while it is itself fully contained within the set of states with negative conditional R{\' e}nyi $2$-entropy. At the same time, we have shown that the conditional amplitude operator criterion is equivalent to the PPT criterion for all Weyl states, but not in general, as we have demonstrated for several examples.

Our article hence provides both an introduction to the topic of entanglement geometry and a step towards the exploration of conditional amplitude operators as general entanglement detection tools. Specifically, it may be of interest for future research to investigate possible generalizations of conditional entropy operators, and to clarify whether the violation of the criterion $\rho\subA{A|B}\leq\mathds{1}$ implies a nonpositive partial transpose in general.

\begin{acknowledgments}
We would like to thank Philipp K{\"o}hler and Heide Narnhofer for fruitful discussions and comments.
\end{acknowledgments}

%\newpage
%\appendix*
%\section{}\label{sec:}

\end{document}